\documentclass[twoside]{article}
\pdfoutput=1
\PassOptionsToPackage{authoryear, round}{natbib}
\usepackage[accepted]{aistats2022}
\usepackage{booktabs} 
\usepackage{xfrac}
\usepackage{algorithm}
\usepackage[noend]{algorithmic}
\usepackage{amsmath,amsthm,amsfonts}
\usepackage{amsmath}
\usepackage{amssymb}
\usepackage{color}
\usepackage{mathrsfs}
\usepackage{verbatim}
\usepackage{enumitem}
\usepackage{bm}
\usepackage[pagebackref]{hyperref}
\usepackage[capitalise,noabbrev]{cleveref}
\usepackage{xr}
\usepackage{tikz}
\usepackage{subcaption}
\usepackage{float}
\usepackage{listings}
\usepackage{array}
\usepackage{makecell}

\usepackage[suppress]{color-edits}
\addauthor{ll}{blue}

\definecolor{DarkBlue}{rgb}{0.1,0.1,0.5}
\hypersetup{
	colorlinks=true,       
	linkcolor=DarkBlue,          
	citecolor=DarkBlue,        
	filecolor=DarkBlue,      
	urlcolor=DarkBlue,          
	pdftitle={},
	pdfauthor={},
}

\newcommand{\cb}{\textcolor[rgb]{1,0,1}{}\textcolor[rgb]{1,0,1}}

\newcommand{\nikhil}{\textcolor[rgb]{0,1,0}{Nikhil: }\textcolor[rgb]{0,0,1}}

\newcommand{\A}{\mathsf{A}}
\newcommand{\G}{\mathsf{G}}

\newcommand{\B}{\mathsf{B}}

\newcommand{\mix}{\textsf{mix}}
\newcommand{\true}{\text{true}}
\newcommand{\pre}{\text{pre}}
\newcommand{\post}{\text{post}}

\newcommand{\potM}{\theta^M}
\newcommand{\potU}{\theta^U}
\newcommand{\Util}{\mathcal{U}}
\newcommand{\cost}{p} 

\newcommand{\swelf}{\mathcal{W}}
\newcommand{\iwelf}{W}
\newcommand{\access}{\mathcal{A}}
\newcommand{\welfgap}{\mathcal{G}}

\newcommand{\soc}{^\mathsf{soc}}
\newcommand{\pri}{^\mathsf{pri}}
\newcommand{\thresG}[1]{\mathsf{thres}_\G(#1)}
\newcommand{\thresA}[1]{\mathsf{thres}_\A(#1)}
\newcommand{\thresB}[1]{\mathsf{thres}_\B(#1)}

\newcommand{\e}[1]{e_{#1}}
\newcommand{\score}[1]{v_{#1}}
\newcommand{\combscore}{v^\alpha}
\newcommand{\f}[1]{f_{#1}}

\newcommand{\argmin}{\mathop{\rm argmin}}
\newcommand{\argmax}{\mathop{\rm argmax}}

\newcommand{\E}{\mathbb{E}}

\newtheorem{theorem}{Theorem}[section]
\newtheorem{definition}{Definition}

\newtheorem{corollary}[theorem]{Corollary}

\newtheorem{remark}{Remark}[section]
\newtheorem{lemma}[theorem]{Lemma}

\usepackage{tabularx}
\usepackage{xcolor}
\usepackage{float,subcaption,placeins}
\usepackage{amsmath}
\usepackage{amsthm}
\usepackage{thm-restate}
\usepackage{float}
\usepackage{enumitem}
\usepackage{xr}
\usepackage{graphicx}
\graphicspath{{./figures/}}
\usepackage{thmtools}
\usepackage{url}
\usepackage{bm}
\usepackage[normalem]{ulem}
\usepackage{balance}
\usepackage{lipsum}
\usepackage{etoc}
\usepackage{comment}

\newcommand{\Comments}{1}
\newcommand{\mynote}[2]{\ifnum\Comments=1\textcolor{#1}{#2}\fi}

\newcommand{\parbold}[1]{\vspace{.25em}\noindent\textbf{#1}}

\usepackage[numbers,sort&compress]{natbib}

\renewcommand\cite[1]{\citep{#1}}

\usepackage[group-separator={,}]{siunitx}
\usepackage{pgfplots}

\pgfmathdeclarefunction{gauss}{2}{%
  \pgfmathparse{1/(#2*sqrt(2*pi))*exp(-((x-#1)^2)/(2*#2^2))}%
}

\usetikzlibrary{patterns}

\pgfplotsset{compat=1.10}
\usepgfplotslibrary{fillbetween}

\usepackage{xcolor}
\definecolor{babyblueeyes}{rgb}{0.63, 0.79, 0.95}
\definecolor{powderblue}{rgb}{0.69, 0.88, 0.9}

\usepackage{chngcntr}
\usepackage{apptools}

\begin{document}

%
\runningtitle{Strategic ranking}

%
\runningauthor{Liu, Garg, Borgs}

\twocolumn[

\aistatstitle{Strategic ranking}

\aistatsauthor{ Lydia T. Liu \And Nikhil Garg \And Christian Borgs }

\aistatsaddress{ University of California, Berkeley
	\And  	Cornell Tech and the Technion \And University of California, Berkeley} 
]

\begin{abstract}
	Strategic classification studies the design of a classifier robust to the manipulation of input by strategic individuals.
	However, the existing literature does not consider the effect of \emph{competition} among individuals as induced by the algorithm design. Motivated by constrained allocation settings such as college admissions, we introduce \textit{strategic ranking}, in which the (designed) individual reward depends on an applicant's post-effort rank in a measurement of interest. Our results illustrate how competition among applicants affects the resulting equilibria and model insights. We analyze how various ranking reward designs\lledit{, belonging to a family of step functions,} trade off applicant, school, and societal utility, \lledit{as well as} how ranking design counters inequities arising from disparate access to resources. In particular, we find that
	randomization in the reward design can mitigate two measures of disparate impact, welfare gap and access.

\end{abstract}

\section{INTRODUCTION}

Many allocative decisions---from education to employment---rely on relative quality across individuals, not absolute quality: Berkeley accepts $\approx 15\%$ of college applicants per year, major CS conferences accept about $\approx20\%$ of submissions, and each job opening is filled by one candidate. Even if every applicant improves substantially, due to (perceived) capacity constraints the number accepted would not commensurately increase. The comparative aspect of \textit{ranking} differs sharply from \textit{classification}, which evaluates each entry in isolation. While this difference is often overlooked in the ML community, we illustrate that it is consequential and motivates \emph{access} as a measure of disparate impact from algorithmic decisions.


Our specific context is the literature on \textit{strategic classification}, which aims to address gaming by applicants controlling the classifier's inputs~\cite{bruckner2012static,Hardt:2016:SC}\lledit{---an instance of Goodhart's law}. For example, if a predictive classifier for admissions learns that a student's number of extracurricular activities correlates with college achievement and thus bases admissions decisions on it, students may list many more activities on their applications without devoting any time to them. A strategic classifier aims to undo this effect, \lledit{informally speaking,} by shifting weight towards features that are costly to game
	, such that the designer's utility is maximized even taking into account strategic behavior---as in a Stackelberg equilibrium.

Importantly, many motivating applications in the strategic classification literature, such as manipulating test scores to gain college admission, are  ranking problems (or classification with capacity constraints on how many can be classified with each label). This aspect induces competition between applicants, affecting their effort and in turn their comparative performance and ultimate relative position in the ranking. To study the interaction of these effects, this work introduces the problem of \textit{strategic ranking}:
{we study the general welfare effects of reward design for ranking, and examine the role of randomization in trading off designer's utility and population welfare, as well as reducing disparate impact.
	We study such \textit{competition} between applications by shifting the task from classifier design to \emph{ranking reward design}, and so fill a gap in the current literature on algorithmic fairness and strategic behavior that has thus far focused on classification (and regression).



\lledit{In pursuit of these questions, our theoretical framework recalls the long-standing economics literature on \textit{contests}, in which agents expend effort to obtain a reward that is a function of their relative performance across agents \cite{bodoh2018college,olszewski2016large,olszewski2019bid,olszewski2019pareto}, as well as that on signaling and strategic behavior \cite{spence1978job}. We elaborate on the connections and distinctions with this literature in Section~\ref{sec:rel}.}

In our strategic ranking framework, a designer ranks applicants by a single, observed measure, called a score. {The score is a function of the applicant's (possibly multi-dimensional) chosen effort level(s), latent skill(s), and environment.} The designer controls the reward $\lambda(\theta)$ assigned to each rank $\theta$, under an overall reward constraint; e.g., in our primary interpretation, the reward is a probability at which an applicant of each rank is admitted. \lledit{Realistically, there may be complexity constraints on the ranking reward function $\lambda$ that the designer can deploy---in this work, we analyze settings where $\lambda$ belongs to a family of step functions that we call ``K-level policies''.} In response, the applicants choose their (costly) {effort level} for the observed measure, such that in equilibrium{---that is, given the effort levels of the entire applicant population---}their effort level and resulting rank maximizes their welfare (reward minus effort cost). Depending {on} the setting, the designer's objective may be a function of applicant effort and who is admitted. \lledit{We assume that applicant effort improves the designer's utility function, but not their own welfare except indirectly through the ranking reward; in other words, applicants are not intrinsically motivated to exert effort.} Using this model, our contributions are:

\textbf{First} (\Cref{sec:modelequil}), we analyze the equilibrium behavior and resulting optimal designs, illuminating important differences between the classification and ranking settings. We show that, in a general setting and with any reasonable reward function $\lambda$, the competition effect results in \textit{rank preservation}: in equilibrium, the ranks (and thus rewards) of applicants after applicant effort is the same as before. This finding differs from strategic classification -- where one must adapt to the effort it induces to maintain accuracy~\cite{Hardt:2016:SC} -- and simplifies equilibrium analysis.

\textbf{Second} (\Cref{sec:welfare}), we study how the design of reward function $\lambda$ differentially affects the welfare of various stakeholders: applicants, a school preferring to admit those with the highest score, and a social planner maximizing the score over the population. We find, e.g., there is an trade-off even in one skill dimension: while deterministically admitting the highest ranked students maximizes the school's utility (among \textit{two}-level policies but not generally), it leads to applicants exerting costly effort. Adding randomness reduces pressure for applicants at the expense of the school's utility.

\textbf{Third} (\Cref{sec:environment})  we analyze equilibria in the presence of structural inequities between different groups, finding that competition amplifies such inequities without careful reward design: designing reward $\lambda$ to increasing school utility increases the welfare gap between groups, and decreases \textit{access} (group specific admission probability) for the disadvantaged group.

Our definitions of school utility and group access---natural in constrained allocation settings---have not been studied in unconstrained settings, where policies may admit different numbers of applicants. We also extend our model to a multi-dimensional setting in \Cref{sec:multi-dim}. 

 \lldelete{Thus, our framework allows for direct comparisons of participant welfare (based on outcomes) at a societal level, beyond ``narrowly bracketed'' notions of algorithmic fairness~\citep{kasy2021fairness}.}



\subsection{Related work}\label{sec:rel}
Our work sits at the intersection of work in two communities,
\textit{strategic classification} and \textit{fair machine learning} in computer science, and \textit{contests} and \textit{effort} in economics. 
%
%
\lldelete{Here, we briefly describe the most related works, with an extended related works section in the Appendix.}




\parbold{Strategic classification} As in strategic classification~\cite{Hardt:2016:SC,bruckner2012static,dong2018strategic}, we consider the challenge an institution faces when deploying a classifier that applicants can game; a naive classifier that does not factor in the resulting distribution shifts would be inaccurate. Unlike strategic classification, which considers the manipulation of observable features to be unproductive from the institution's perspective, our model of efforts stipulates that the school prefers higher effort levels as they give rise to higher scores.  We follow the line of work considering the design of reward functions that further incentivize agent effort on productive tasks~\cite{kleinberg18investeffort, miller2020strategic,bechavod2021gaming,shavit2020causal}. We in particular draw inspiration from and compare to several recent directions in the strategic classification literature: with multiple agents \citep{haghtalab2020maximizing,alon2020multiagent}, when a classifier may be random \citep{braverman2020role}, and with fairness concerns \citep{hu2019disparate,Milli2019social}. Like other strategic classification work, our work is further related to the prior economics literature on contract design, elaborated in \Cref{sec:multi-dim}. 

\parbold{Economics: Contests and modeling effort} Our theoretical analysis is similar to that in the literature on \textit{contests} in economics~\cite{barut1998symmetric,konrad2007strategy,connelly2014tournament, bodoh2018college,olszewski2016large,olszewski2019bid,olszewski2019pareto,fang2020turning}; e.g., our rank preservation result in \Cref{sec:equillemmas} reflects assortative allocation results in the literature under similar conditions. The field is too extensive to summarize here, so we refer the reader to surveys~\cite{corchon2007theory,fu2019contests} and discuss the works closest to ours.
\citet{bodoh2018college} develop a model with students endogenously choosing effort; they use the model to compare various affirmative action schemes. 
Most related is work by~\citet{olszewski2016large,olszewski2019bid,olszewski2019pareto}. \citet{olszewski2019pareto} find that policies that ``pool'' individuals into tiers (as opposed to continuous ranks) can be Pareto improving for students; even those with decreased individual rewards would benefit from decreased competitive pressure.

Our work departs from the contests literature through its focus on the questions most common in the strategic classification: while the former primarily considers the efforts and resulting welfare of participants, we study how the design of the reward function $\lambda$ differentially affects the designer's utility, applicant welfare, and fairness metrics. In particular, when $\lambda$ is interpreted as a probability of admission, our designer faces trade-offs with incentivizing overall effort and admitting the most skilled students (\Cref{sec:welfare}); with disparate access, the designer must further navigate trade-offs between fairness and inducing effort (\Cref{sec:environment}); with multiple score dimensions, the designer must ensure that applicants do not excessively game one dimension at the cost of the other (\Cref{sec:multi-dim}). Simultaneous to us, \citet{elkind2021contest} consider the effect of contest design on participant and designer welfare.

Beyond contests, our work connects to the economics literature on effort and subsequent reward \citep{becker1973theory, spence1978job,roemer1998equality,calsamiglia09decent}. 

\parbold{FATE in machine learning and MD4SG} We broadly connect to the Fairness, Accessibility/Accountability, Transparency, and Ethics in machine learning \cite{Chouldechova2018TheFO} and the Mechanism Design for Social Good literatures \cite{abebe2018mechanism}. Most relevant are works on fair ranking \cite{mathioudakis2020affirmative,zehlike2017fa,zehlike2020reducing,tabibian2020design}, constrained allocation \cite{aziz2020developments,caifair2020,golz2019paradoxes,noriega2019active}, and admissions \cite{faenza2020impact,garg2020dropping,immorlica2019access,kannan2021best,liu2020disparate,Hu2018shortterm,Mouzannar2019socialequality,kannan2019downstream,liu2018delayed,rolf2020balancing,liu2021test}. We consider such effects as they interact with agents' strategic responses to the mechanism.

\section{\uppercase{Model and equilibrium}}
\label{sec:modelequil}
Each \textit{applicant} has an (unobserved) skill level and so a pre-effort \textit{rank}. \lledit{The applicants are modeled as a continuum of players \cite{Schmeidler1973}, rather than a finite number of atoms; informally, applicants do not respond to the strategies of every other applicant combinatorially, but rather the applicant population as a whole. }
	 Applicants choose effort levels, resulting in post-effort scores and ranks. A single \textit{school} determines rewards for each post-effort rank level, thus affecting applicant incentives to choose their effort. 


\subsection{Model}
\label{sec:model}


\parbold{Applicants}
There is a unit mass of \textit{applicants}, indexed by an observed index $\omega \in [0,1]$ distributed uniformly.\footnote{The index $\omega$ should be interpreted as each applicant's ``name,'' uncorrelated with skill, used solely for tie-breaking.}  Each applicant has a latent (unobserved) skill level represented by some measurable function of $\omega$.
We assume that {the distribution of the skills} has no atoms and that the CDF of this distribution is strictly increasing.
Using the CDF to map the skill of an applicant to a rank in~$[0,1]$, each applicant gets an (unobserved) rank $\theta_\pre =\theta_\pre(\omega)$ which by our assumption on the CDF is again uniformly distributed in $[0, 1]$ (the higher the rank the better).  With this setup, the skill of an applicant with rank $\theta_\pre$ can be written as $f(\theta_\pre)$ where $f$ is a strictly increasing, continuous function.  {It will be notationally convenient to label applicants by their rank $\theta_\pre$, though the reader should note that in contrast to $\omega$, $\theta_\pre(\omega)$ is assumed to be unobservable.}


Each applicant chooses an \textit{effort} level $e \ge 0$, the result of which is an observed, post-effort \textit{score}, {$v=v(e,\theta_\pre) = {g(e)}\cdot{f(\theta_\pre)}$}. \lledit{In other words, we assume the post-effort score is the product of two components, associated with the effort and the pre-effort skill respectively.}
{We assume} that the effort transfer function $g: [0, \infty) \mapsto [0, \infty)$ is a continuous, concave, strictly increasing function, representing that marginal effort improves one's score but has diminishing returns.
{
{T}he strategies of the applicants can {then} be described by a function $\theta_\pre \mapsto e(\theta_\pre)$.
}
  Each {applicant} is then ranked according to their score $v$
, resulting in a \textit{post-effort rank} $\theta_\post$.
Note that the ranking $\theta_\post$ is slightly less trivial than a ranking of the skills, since the scores might have ties, which have to be resolved.

{\textbf{Tie Breaking.}
Given a choice of strategies $\theta_\pre \mapsto e(\theta_\pre)$, let $F$ be the CDF of the scores $v(e(\theta_\pre),\theta_\pre)$.
Since atoms for the distribution of $v$ would lead to ties for ranks defined as $F(v)$, we will use the labels of the applicants to break ties with the help of a (publicly announced) tie-breaking function $\Gamma(\omega)$, defined, e.g., via a collision free hash of the applicant names. Here  we require  that $\Gamma$  is a measurable function from $[0,1]$ to $[0,1]$ that maps different applicant labels to different values.
We then use $\Gamma$ to resolve the atoms of $F$, leading to a ranking function $v\mapsto \gamma(\omega, v)$ which is equal to $F(v)$ except when $v$ is an atom of the score distribution, {in which case it takes values in the ``gap interval'' $[F_-(v),  F(v)]$, where $F_-(v)$ is the left limit of $F$ at $v$.}  We  construct $\gamma$ in such a way that it
gives the uniform distribution for $\theta_\post(\omega)$ if we set  $\theta_\post(\omega)=\gamma(\omega,v(e(\theta_\pre(\omega)),\theta_\pre(\omega)))$.\footnote{See Remark~\ref{rem:gamma} for details on this construction.}
}



\parbold{Ranking designer (school)} A single \textit{school} is admitting applicants, based on their ranking. In particular, the school can choose a ranking reward function $\lambda : [0, 1] \mapsto [0, 1]$, such that an applicant with post-effort rank $\theta_\post$ is admitted with probability $\lambda(\theta_\post)$. We assume that $\lambda$ is non-decreasing and that the school has a constraint on the overall probability, such that in expectation it admits a number of applicants equal to {a} capacity constraint  $\rho \in (0,1)$, i.e., \lledit{$\E_{\theta_\post}[\lambda(\theta_\post)] = \rho $}.\footnote{\lledit{We use $\E_{\theta_\post}[\cdot]$ to denote an integral over $\theta_\post$ (and $\E[\cdot]$ to denote an integral over $\omega$)} with respect to the Lebesgue measure; informally, this can be thought of as averaging over the applicant population.} We may also refer to $\lambda$, informally, as the admission policy.

For simplicity, we further assume that $\lambda$ is a step-function with $K$ distinct levels $\ell_0 < \dots < \ell_{K-1}$, and $K-1$ cut-points parameterized by $c_1< \dots < c_{K-1}$ (with $c_0 = 0$, $c_K = 1$). 
In other words, we have $\lambda(\theta) = \ell_k$, for all $\theta \in \psi_k \triangleq [c_{k}, c_{k+1})$. Thus, applicants in the same post-rank interval $\theta \in \psi_k$ receive the same reward.

%
%
%

\parbold{Individual applicant welfare and equilibrium} Given the designer's function $\lambda$ and the effort levels of other applicants, each applicant chooses effort $e$ to maximize their individual welfare, 
\begin{equation*}
	\iwelf(e, \lambda(\theta_\post)) = \lambda(\theta_\post) - \cost(e),
\end{equation*}
where the effort cost function $\cost$ is non-negative, continuous, and strictly convex {on $[0,\infty)$, with $e_0:=\argmin_e \cost(e)$ and $p(e_0) = 0$.}

Applicants are assumed not to personally benefit from increasing their score $v$, except through the corresponding increase in their rank and reward. While the definition of $\cost$ does not preclude the applicant receiving any intrinsic benefit from exerting effort, we assume that the \emph{net} benefit from effort is non-positive.

After a school chooses its ranking reward function $\lambda$, each applicant chooses their effort level. However, unlike in strategic classification, in the ranking setting applicants must further take into account the effort levels (and resulting post-effort values) of other applicants. 
An \textit{equilibrium} of effort levels is
{then an assignment $\theta_\pre\mapsto e(\theta_\pre)$ of effort levels} and resulting post-effort rank rewards
in which given the efforts of other applicants, no applicant can increase their welfare by changing their effort.  
{This is formalized in the following definition.}

\begin{definition} [Equilibrium]\label{def:equi}
{Given a}
tie-breaking function $\Gamma(\omega)$
%
{and}
	a ranking probability function $\lambda$, an \textit{equilibrium} is a set of effort levels and post-effort ranking rewards for each applicant, $\{e(\theta_\pre),  \lambda(\theta_\post(\theta_\pre))\}$ such that, for all
	$\omega$, 
	\begin{align*}
		e(\theta_\pre(\omega)) &\in \argmax_{e} \iwelf\left(e, \lambda\left(\gamma\left(\omega, v\left(e,\theta_\pre(\omega)\right)\right)\right)\right) 
		\\
		\theta_\post(\theta_\pre(\omega)) &= \gamma(\omega, v(e(\theta_\pre(\omega)),\theta_\pre(\omega))),
	\end{align*}
	{where $\gamma(\omega,v)$ is the ranking induced by the CDF of the scores resulting from effort levels $\{e(\theta_\pre)\}$ and the tie-breaking function $\Gamma$.\footnote{See Remark~\ref{rem:gamma-effort} on $\gamma$ and the set of efforts.
	}
	}
\end{definition}
{In equilibrium, the strategy of the applicants is
{thus} characterized by a collection of efforts $e(\theta_{pre})$ and post-effort ranking rewards $\lambda(\theta_{post}(\theta_{pre}))$ of the applicants as a function of $\theta_{pre}$, with a corresponding joint distribution for $e$ and $\theta_{post}$ induced by the equilibrium and the underlying uniform distribution over $\theta_{pre}$.\footnote{See Remark~\ref{rem:design-opt} on credible commitment by the school.
}
} Intuitively, we need to define index $\omega$ and the tie-breaking ordering $\Gamma$ because while there may be ties in post-effort values, post-effort \textit{ranks} must be unique. However, for notational ease, in the rest of this work we drop the index $\omega$ and instead refer to applicants by their pre-effort rank $\theta_\pre$. Our results hold for any ordering $\Gamma$, and so we further omit it.


\parbold{Aggregate welfare and utility} We define three aggregate welfare functions of the equilibrium efforts and scores to capture the interests of different stakeholders.
The \textbf{applicant welfare} is defined as the population average of the individual applicant welfare at equilibrium:
\begin{align*}
	\swelf :=& \E[\iwelf(e,\lambda(\theta_\post))]	= \rho - \E[\cost(e)]. 
\end{align*}
On the other hand, society derives value from the scores of applicants post-effort, leading to the following \textbf{societal utility}:$\quad\quad  	\Util\soc:= \E[ v].$

In other words, society prefers the entire applicant population to achieve higher scores, not only those who are admitted to the school, since higher test scores are correlated with labor productivity and economic growth \citep{hanushek2010high}.

Finally, the design $\lambda$ is controlled by a school, who may only draw value from those who enroll. The school's \textbf{private utility} is the expected score of admitted applicants,
which in our continuum formulation is the expectation of $v$ weighted by $\lambda(v)$: 
\(	\Util\pri:= \E[v\cdot \lambda(\theta_\post)]\).


\parbold{Discussion} Our three welfare functions represent the utilities of three stakeholders in any ranking setting: the applicants, ranking institution, and broader society. As we show in Section~\ref{sec:welfare}, optimal design differs substantially for the three -- the school maximizing its own utility comes at a cost to the others. Furthermore, we note that this comparison is possible \textit{because} our setting is one of constrained allocation, in which the expected number of admitted applicants (integral over $\lambda$) is fixed. For example, defining school's private utility as the expected score of admitted applicants (and applicant welfare as containing admissions probability) is unnatural without such a constraint, as various mechanisms may admit (classify as `1') different numbers of applicants; for this reason, strategic classification papers often consider classification accuracy as the school's objective, which has no analogue in practice. \lldelete{We believe that this shift is an important aspect of considering the welfare implications of mechanisms beyond ``fairness'' characterized by statistical notions, as, e.g., recently advocated by~\citet{kasy2021fairness}.}


Our base model is purposely minimalist, to emphasize the competitive aspect of ranking and constrained allocation. We consider extensions to disparate access to resources between socioeconomic groups (\Cref{sec:environment}) and multi-dimensional, potentially unobserved scores (\Cref{sec:multi-dim}) after studying welfare under our base model.



\subsection{Equilibria characterization}
\label{sec:equillemmas}

The definition of equilibria suggests that studying their properties may be difficult: our utility functions depend on both the efforts of applicants and their induced relative rankings, the relationship of which may be complex in general.
%
 However, the following result, reminiscent of much of the contests literature under similar assumptions (see Remark~\ref{rem:contests}), establishes that
 ranking rewards are preserved under effort.

\begin{restatable}[Rank preservation]{proposition}{lemassortative}
	\label{lem:rankpreservedindex}
In every equilibrium, $ \lambda(\theta_\post(\theta_\pre))=\lambda(\theta_\pre)$, up to 
{sets of} measure 0. 
\end{restatable}

As shown in the next section, rank preservation simplifies substantially the evaluation of various utility functions, as the post-effort rankings of applicants are fixed and known.

The result follows from the shared cost function $\cost$ being convex and increasing on $[e_0, \infty)$: if an applicant pre-effort rank $\theta$ finds it optimal to achieve post-effort value $v$, then each applicant with pre-effort rank $\bar\theta > \theta$ finds it optimal to reach a post-effort value $\bar v\ge v$. We note that this result (and all our subsequent results) hold for any tie-breaking function 
{$\Gamma$,}
as any set of post-effort values in which tie-breaking occurs across two or more reward bands $\psi$ cannot constitute an equilibrium. 

\begin{figure}[t]
	\centering
	\includegraphics[width=0.45\textwidth]{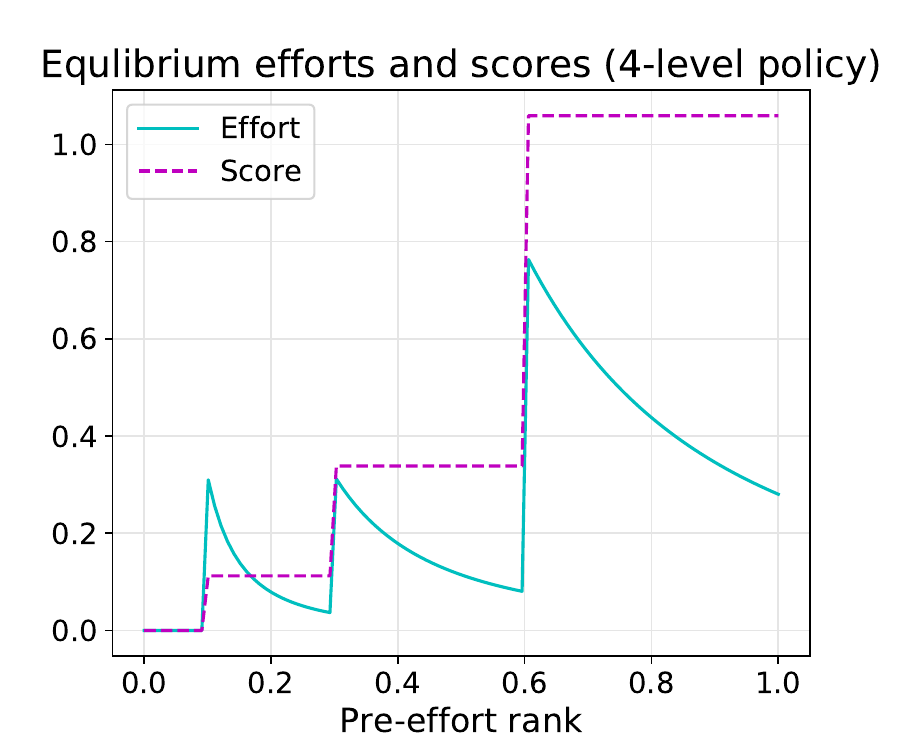}
	\caption{Equilibrium efforts and scores for a sample 4-level ranking reward function.}\label{fig:eqm}
\end{figure}

%

%

While the ranking reward function $\lambda$ does not affect the induced rankings of applicants, it does determine the \textit{effort} exerted by each, as formalized next.

\begin{restatable}[Second price effort]{theorem}{lemeffort}
	\label{lem:effort}
	There exists a equilibrium such that $\lambda(\theta_\post(\theta_\pre))=\lambda(\theta_\pre)$ and applicants with $\theta_\pre \in \psi_k\triangleq{ [c_{k}, c_{k+1})}$ exert effort $e_k(\theta_\pre)$, where 
	\begin{align*}
		e_k(\theta_\pre) &= \begin{cases}
			e_0 & \text{for } k = 0\\
			\max\left(g^{-1}\left(\frac{g(\tilde e_{k-1})\cdot f(c_{k})}{f(\theta_\pre)} \right), e_0\right) & \text{otherwise,}
		\end{cases}
	\end{align*}
	\label{lem:secondpricecontimuum}
	with $\tilde e_{k-1}{\geq e_{k-1}(c_k)}$ {inductively defined by} $$
	{\cost(\tilde e_{k-1})=\cost(e_{k-1}(c_k))+\ell_k-\ell_{k-1}}
		$$
		The equilibrium is unique up to sets of  measure 0. 
\end{restatable}


As the theorem name suggests, the effort exerted by each applicant in equilibrium is akin to the price paid in a second price auction (cf. \citet{myerson1981optimal}): each applicant exerts just enough effort that applicants in the level below (those with ranks at or below the given applicant's level's lower cut-point $c_k$) cannot increase their welfare by instead exerting additional effort, $\tilde e_{k-1}{>e_0}$.  {It is instructive to calculate the post-effort scores resulting from these efforts; for $\theta_\pre\in \psi_k$, they are
	\begin{align*}
		v_k(\theta_\pre) &= \begin{cases}
			g(e_0)f(\theta_\pre) \quad \quad \text{for }
			k = 0\\
			\max\Big\{{g(\tilde e_{k-1}) f(c_{k})}\, ,\, g(e_0)f(\theta_\pre)\Big\} ~ \text{o.w.,}
		\end{cases}
	\end{align*}
Thus at the beginning of each band, applicants exert strictly more effort that those at the top of the previous band, decreasing their effort within the band since their increased skill requires less effort to get the same score.  If at some point effort $e_0$ (with cost $\cost(e_0)=0$) is enough to maintain the score needed to get reward $\ell_k$, their effort stays constant and their score grows, up to the beginning of the next band, when at least initially, the scores again stay constant and the efforts decrease.}
Figure~\ref{fig:eqm} illustrates the equilibrium efforts and post-effort scores under a sample 4-level reward function, and choices of $f, g$ and $\cost$ such that it is never enough to just exert effort $e_0$ except in the first band.\footnote{ These were: $f(x) = 2x$ (corresponding to a uniform distribution of skill levels), $g(x) = \sqrt{x}$ (decreasing marginal returns to effort) and $\cost = x^2$ (increasing marginal cost of effort).}

Thus, while {the} reward function $\lambda$ does not change the ranking of applicants (by Lemma~\ref{lem:rankpreservedindex}), it does affect the effort exerted and thus their post-effort scores. The function $\lambda$ is thus a design parameter for the school whose objective depends on applicant scores, as we explore in the next section.\footnote{See Corollary~\ref{corr:effectcomparativestatics} for an example of how equilibria change with {the} function $\lambda$}


Remark~\ref{rem:contests} details the technical differences between our results and analogues in the contests literature \citep{bodoh2018college,olszewski2016large}.

\section{\uppercase{Welfare analysis}}
\label{sec:welfare}

We begin our analysis by studying how the design of ranking reward function $\lambda$ changes the utilities of the applicants, the school, and the society. \lldelete{As shown in the previous section, the reward $\lambda$ influences the applicant equilibrium effort levels and thus the distribution of their scores. Moreover, by definition, $\lambda$ indicates the admissions probability given a rank, and so changes in $\lambda$ further affect the rank composition of those admitted. These multiple effects induce tradeoffs between the various welfare notions, as formalized in this section. The design of $\lambda$ should navigate these tradeoffs.}

Recall that the school is constrained to a $K$-level $\lambda$ admission policy such that the average probability of admission is $\rho > 0$, that is, $ \E[\lambda(\theta_\post)] = \rho$. A special case of interest is the following \textit{two-level} ($K=2$) function.


\begin{definition}[Two-level policy]\label{assump:2-level-fixed-cap} In our baseline two-level function parameterized by cut-off $c \in (0,1-\rho]$, each applicant with post-effort rank  $\theta_\post \geq c$ is admitted with probability $\ell_1 =\frac{\rho}{1-c} \in (0,1]$. Others are rejected, $\ell_0 = 0$.
\end{definition}

Note that standard non-randomized admissions policies are equivalent to case where $c = 1 - \rho$: the highest scoring applicants up to the capacity constraint are accepted with probability 1 and all others are rejected. We call this case \textbf{non-randomized admissions}. The other extreme is the one-level policy, \textbf{pure randomization}, where each applicant is admitted with probability ${\rho}$, i.e., $\ell_0 =\rho$. Decreasing the cut-off $c$ can be viewed as increasing the level of randomization in the admissions policy.

{We now reason about how various welfare and fairness metrics vary with $\lambda$ (in the two-level policy class, just $c$). To simplify the presentation, we assume that $f, g$ and $\cost$ are differentiable and that baseline effort $e_0$ is~$0$.}







\parbold{Applicant welfare} The following result shows how overall welfare is maximized. 

\begin{restatable}[Applicant welfare]{proposition}{propStudentWelfare}\label{prop:student-welfare}
	Among all $\lambda$ with $K$ levels, for $K\geq 1$, applicant welfare $\swelf$ is maximized by the one-level policy with pure randomization. Further, in the class of two-level policies, $\swelf$ is monotonically non-increasing in $c$.
\end{restatable}

This result is perhaps unsurprising given our formulation of applicant welfare -- in which effort is costly but applicants do not directly benefit from their score $v$, only their resulting admissions probabilities $\lambda(\theta_\post)$. Thus, $\swelf$ is maximized by a completely random admissions policy where the cost of exerting effort is $\cost(e_0) = 0$ for every applicant.

\parbold{School's private utility}\label{sec:private-util-max}
On the other extreme is the school's private utility $\Util\pri$: in which only the expected scores of admitted applicants matters. For $K=2$, it is maximized by a deterministic decision policy.

\begin{restatable}[School's private utility for two-level policies]{proposition}{propPrivateUtil}
\label{prop:private-util-max}
In a two-level policy, the school's private utility~$\Util\pri$ is monotonically non-decreasing in $c$, and consequently is maximized by non-randomized admissions.
\end{restatable}

The result follows because both aspects of the school's private utility increases with the cutoff $c$ in a two-level policy: first, as $c$ increases, the school admits higher ranked applicants and the scores of the admitted applicants increase with the rank of the admitted applicants; second,  increasing $c$ increases $\ell_1$, which, as per \Cref{corr:effectcomparativestatics}, further increases the equilibrium effort levels (and thus the post-effort scores) of the highest ranked applicants. Note that this effect occurs even though rankings of applicants are identical under any $\lambda$ (\Cref{lem:rankpreservedindex}).


A natural question is whether the deterministic decision policy also maximizes $\Util\pri$ among all $K$-level policies for $K>2$. Surprisingly, the answer in general is negative.

\begin{restatable}[]{proposition}{counterexPrivateUtil}\label{prop-counterex-priv-util}
	 A $3$-level policy may achieve strictly higher $\Util\pri$ than non-randomized admissions.
\end{restatable}
This result provides a perhaps counter-intuitive insight for practice: \textit{even for a school maximizing its own utility}, deterministically accepting the top students is not generally optimal -- schools could improve on both student welfare and its own utility by randomizing.   
The counter-example used to prove the above proposition involves picking a skill distribution (the distribution of $f(\theta_\true)$) with a long tail. That is, if the skill level $f(\theta_\pre)$  of applicants with rankings above the optimal deterministic cutoff $1-\rho$ is relatively high, then the school can improve $\Util\pri$ by using a three-level policy that increases the competition for high admission probability and incentivizes higher scores among the top ranked applicants.

\parbold{Societal utility}
The previous results show that for both applicant welfare and school's private utility, an extreme two-level function is optimal among all two-level admissions functions $\lambda$: pure randomization for applicant welfare, and non-randomization for school's private utility. Next, we show that a similar result does not hold for societal utility, $\Util\soc = \E[v]$: if the goal is to maximize the score attained by the entire applicant population and not just the admitted applicants, the admissions function $\lambda$ should be randomized (but not purely randomized), even among two-level policies.\footnote{Our notion of societal utility
	differs from the traditional notion of \emph{social welfare}, which would also include in its formulation the effort costs of applicants. \Cref{prop:student-welfare} and \Cref{prop:SocUtil} would imply that social welfare is also maximized by an intermediate level of randomization.}

\begin{restatable}[Societal utility maximization for two-level policies]{proposition}{propSocUtil}
	Among two-level functions, there exists a setting in which societal utility $\Util\soc$ is maximized by choosing $c \in(0, 1-\rho)$, $\ell_1 = \frac{\rho}{1-c} \in (\rho,1)$.
\label{prop:SocUtil}
\end{restatable}

Here, increasing $c$ and $\ell_1$ has two competing effects: (1) 
  it increases the effort level of the highest ranked applicants in $\psi_1$; (2)  it simultaneously increases the fraction of applicants who exert minimum effort, as they go from being admitted with probability $\ell_1$ (and thus needing to exert effort to preserve their position) to being deterministically rejected. Maximizing overall score across the population thus requires an intermediate level of randomization. More generally, a $K$-level reward function for $K> 2$ may improve societal utility over a two-level policy.

The choice of $\lambda$ not only differentially affects applicants' admissions probabilities; it also differentially affects applicant \textit{efforts} and \textit{scores} in equilibrium. Choosing a reward function $\lambda$ to increase the equilibrium scores of some applicants comes at a cost of the scores of other applicants. Thus, as the contrast between $\Util\soc$ and $\Util\pri$ establishes, the optimal design depends on which applicants one considers. When a school has the power to choose design $\lambda$ to maximize its own private interests, doing so comes at a cost of scores of other applicants. As shown in Figure~\ref{fig:welfare-tradeoffs}, the exact tradeoffs between $\Util\soc$, $\Util\pri$, and applicant welfare $\swelf$ depend on various model parameters, such as $f$, $g$ and $\cost$, as well as the capacity~$\rho$.

\begin{figure}[tbp]
	\centering
	\includegraphics[width=0.45\textwidth]{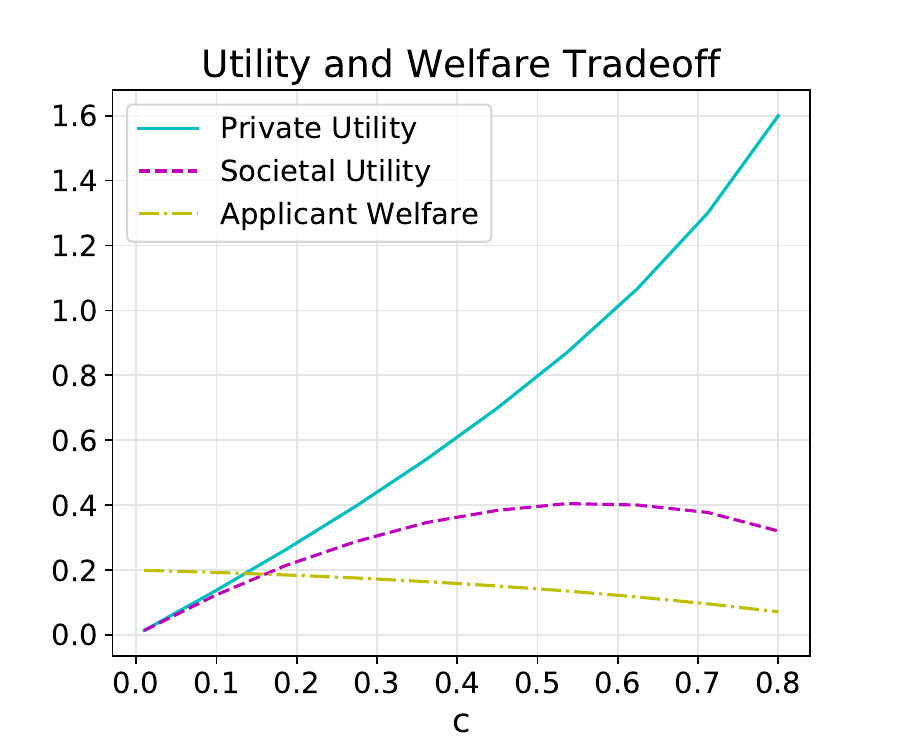}
	\caption{$\Util\soc$, $\Util\pri$, and $\swelf$ for two-level policies parametrized by $c \in (0,1-\rho)$.  Model parameters are as in Figure~\ref{fig:eqm}. The school's capacity is $\rho = 0.2$.
	}\label{fig:welfare-tradeoffs}
\end{figure}

\section{\uppercase{Environment differences}}\label{sec:environment}

We have thus far assumed a basic model where applicants differ only in their latent skill levels $f(\theta_\pre)$, and their observed score depends only on their latent skill level and chosen effort. However, in societies with structural inequalities, an individual's measured success also depends on various environmental factors beyond one's control, such as family income and the availability of resources in one's community. \citet{roemer1998equality}'s work on the equality of opportunity argues that environmental factors (called ``circumstances'') are distinguished from effort. For example, for the same amount of time spent studying, an applicant from a well-resourced school may achieve higher test scores than one from an under-resourced school. In this section, we extend the model introduced in Section~\ref{sec:modelequil} to study the disparate impact of admission policies in the presence of structural inequalities, specifically differences in the applicant's previous education environment. 

\paragraph{Model and equilibria characterization} We now denote each applicant's latent skill rank as $\theta_\true \in [0,1]$. In addition to the latent skill, each applicant has an (unobserved) environmental factor $\psi \in \Psi$ that represents how favorable their environment is for attaining a higher score. Because a favorable environment results in a higher rate of return for effort, we model the environment as a multiplicative factor in the score (see~e.g., \citet{calsamiglia09decent}). Formally, the post-effort score is a function of the latent skill, the environmental factor, and the effort level:
\begin{equation*}\label{eq:env-score}
v = \psi \cdot  g(e)\cdot f(\theta_\true),
\end{equation*}
where $g, f$ are as defined in Section~\ref{sec:model}. 


 We assume there are two groups of applicants, $\A$ and $\B$, and the distribution of skill is the same in both groups. Group $\A$ has a more favorable environment factor, that is,  $\Psi = \{\psi_\A, \psi_\B\}$ and $\psi_\A > \psi_\B$. Thus we will also refer to $\B$ as the ``disadvantaged group". To simplify our presentation, we assume each group is half of the total applicant population, though the results in this section generalize.
We defer all proofs in this section to Appendix~\ref{app:environment}.

We begin by characterizing the equilibrium ranking under the designer's policy $\lambda$. The equilibrium effort levels and post-effort ranks are as defined in Definition~\ref{def:equi}, except they are now group-dependent, that is, we have $e(\theta_\true, \psi)$ and $\theta_\post(\theta_\true, \psi)$. Because of the differences in $\psi$, the post-effort ranking $\theta_\post$ is now group-dependent, and in general is not equal to $\theta_\true$. To apply Proposition~\ref{lem:rankpreservedindex} as before, we construct an ``environment-scaled pre-effort rank'' (denoted $\theta_\pre$).

\begin{restatable}[Equilibrium under group differences]{proposition}{propEquiEnv}\label{prop:equi-env}
Define $\theta_\pre$ as:
	\(\theta_\pre := f^{-1}_\mix(f(\theta_\true)\cdot \psi)\),
		where $f^{-1}_\mix$	is the CDF for the environment-scaled skill, $f(\theta_\true)\cdot \psi$:
	\begin{equation*}
	f^{-1}_\mix(x) := \frac{1}{2}f^{-1}(x/\psi_\A)+ \frac{1}{2}f^{-1}(x/\psi_\B).
	\end{equation*}

	Then, in every equilibrium, $\lambda(\theta_{\post}(\theta_\true, \psi))=\lambda(\theta_\pre)$.
\end{restatable}


The environment factor depresses the $\theta_\pre$, and therefore $\theta_\post$, of the disadvantaged group. Rank preservation (\Cref{lem:rankpreservedindex}) again simplifies equilibrium analysis. 

\paragraph{Welfare gap and access differences}
Using the above characterization, we now study the disparate impact of various rank reward functions $\lambda$,
through two key facets: (1) the \emph{welfare gap} and (2) \emph{access}. The former measures welfare disparity between applicants from different groups \emph{with the same latent skill rank}.

\begin{definition}[Welfare gap]
	Let $\swelf^\G(\theta_\true)$ denote post-effort welfare of an applicant with latent skill ranking $\theta_\true$ from group $\G \in \{\A, \B\}$, i.e.,
	\begin{equation}
	\swelf^\G(\theta_\true):= \lambda(\theta_\post(\theta_\true, \psi_\G)) - \cost(e(\theta_\true, \psi_\G)).
	\end{equation}
	We define the \emph{
		welfare gap} as 
	\begin{equation*}
		\welfgap(\theta_\true):= \swelf^\A(\theta_\true) - \swelf^\B(\theta_\true).
	\end{equation*}
\end{definition}

The welfare gap captures differences in admission probabilities \textit{and} in the effort needed to achieve such probabilities. Our next notion, \textit{access}, captures whether a decision policy \emph{includes} the disadvantaged group in the admitted class, regardless of effort. 

\begin{definition}[Access]
 \emph{Access} is the overall probability of admission of the disadvantaged group.
	\begin{equation*}
		\access :=  \E_{\theta_\true}[\lambda(\theta_\post(\theta_\true, \psi_\B))].
	\end{equation*}
\end{definition}


\begin{restatable}[Admission and pointwise welfare gap for two-level policies]{proposition}{propEnvDiff}\label{prop:environment-diff}
	Denote the group-specific rank threshold for group $\G$ as
	\begin{equation*}
	\thresG{c} := f^{-1}\left(\frac{f_\mix(c)}{\psi_\G}\right).
	\end{equation*}
A two-level policy with $c \in (0,1-\rho]$ admits a group~$\A$ applicant with $\theta_\true \geq \thresA{c} $ with probability $\frac{\rho}{1-c}$, a group~$\B$ applicant with $\theta_\true \geq \thresB{c} $ with probability $\frac{\rho}{1-c}$, and all other applicants with probability~0. 
The welfare gap $\welfgap(\theta_\true)$ is non-negative for every $\theta_\true$, and strictly positive for $\theta_\true \geq \thresB{c}$. In contrast, the one-level pure randomization policy has $\welfgap(\theta_\true) \equiv 0$.

\end{restatable}

In the ``High" region of $\theta_\true$ (where $\theta_\true \ge \thresB{c}$), both the group $\A$ and group $\B$ applicant have the same probability of admission. However, group $\B$ applicants must expend more effort, resulting in a strictly positive welfare gap. In the ``Middle'' region (where $\theta_\true \in [ \thresA{c}, \thresB{c})$), group $\A$ applicants are admitted with positive probability while group $\B$ applicants are deterministically rejected, again leading to a positive welfare gap. In the ``Low'' region (where $\theta_\true < \thresA{c}$), applicants from both groups are deterministically rejected, leading to no gap. Further comparison of the equilibrium welfare, admission probability and effort for any given two-level policy can be found in Table~\ref{tab:comparison} in the appendix.

The previous proposition highlights the joint role of admissions probability and effort in determining the welfare gap. Our next result focuses on the welfare gap in the ``High" region, where applicants of both groups are admitted with the same positive probability, showing that the welfare gap decreases as a two-level admission policy becomes more randomized.

\begin{restatable}[Welfare gap increases with $c$]{proposition}{propWelfareGap}\label{prop:deriv-welfare-gap}
	Consider the setting in Definition~\ref{assump:2-level-fixed-cap}, with the school's chosen admissions policy $c=\bar{c}$, where $\bar{c} \le  (f_\mix)^{-1}(f(1)\cdot \psi_\B)$. Then, for any $\theta_\true \geq \thresB{\bar{c}}$, that is, $\theta_\true$ is in the ``High'' region, we have
	\(\left.\frac{\partial \welfgap(\theta_\true)}{\partial c}\right\vert_{c=\bar{c}}> 0\).
\end{restatable}


Decreasing $c$ (increasing randomization) reduces the welfare gap in the ``High'' region of $\theta_\true$. As the school's private utility is actually increasing in $c$ (Proposition~\ref{prop:private-util-max}), there is also a tradeoff between the school's private utility and the welfare gap.

Increasing randomization also increases access: it is maximized by pure randomization and there exists a large class of models where access is always improved by more randomization in a two-level policy.

\begin{restatable}[Access decreases with $c$]{proposition}{propAccess}
	Pure randomization has higher $\access$ than any two-level policy.
	Moreover, if $f^{-1}$ is convex, $\access$ for  two-level policies is non-increasing in $c$. 
	\label{prop:access}
\end{restatable}

By reducing the level of competition among applicants, the more randomized admission policy increases access, becoming more inclusive to applicants who are disadvantaged by their environment.
A ranking reward function that increases the incentive for applicants to compete tends to increase the welfare gap and reduce access. Without knowledge of the applicant's environment factor, randomizing the admissions policy thus may be key to reducing disparate outcomes.

%
%
%
%
%
%

\section{\uppercase{Discussion}}
\label{sec:conclusion}

We now draw some comparisons to welfare results known in strategic classification. Taken together, Proposition~\ref{prop:student-welfare} and~\ref{prop:private-util-max} imply, among two-level policies, that there is a direct tradeoff between the applicant welfare and the school's private utility: by reducing the degree of randomization in the admission policy, we increase private utility at the expense of applicant welfare. In the strategic classification setting, \citet{Milli2019social} observed a similar tradeoff between the ``institution utility'' (the classification loss under gaming) and the ``social burden'' (cost incurred by individuals for changing their features). Both sets of results suggest that strategic behavior complicates the choice of an optimal decision policy and necessitates careful adjudication among stakeholders.

Our key takeaway that randomization in the decision policy $\lambda$ can improve applicant welfare 
is also related to \citet{braverman2020role}. They studied the welfare benefits of randomization---in the form of probabilistic classifiers and noisy features---and observed that the designer has no incentive to use a more randomized classifier. This finding is true in our setting, only \textit{for the class of two-level policies} -- randomization with more levels may improve both welfare and private school utility.  

The strategic classification literature \citep{hu2019disparate,Milli2019social,braverman2020role} also considers the disparate costs of strategic behavior. These works study two groups with different costs of \emph{gaming}, and \citet{Milli2019social} introduces the concept of ``social gap'': the difference between the costs of successful gaming incurred. Our \emph{welfare gap} can be seen as a measure of the difference in the costs of effort in the ``High" region of latent skill rank. We showed a tradeoff between the school's private utility and the welfare gap (Proposition~\ref{prop:private-util-max} and \ref{prop:deriv-welfare-gap}), while \citet{Milli2019social} (Theorem 3.1) showed a tradeoff between the institutional utility and the social gap. On the other hand, the concept of access---the proportion of the disadvantaged group admitted---as an indicator of disparate impact has not received much attention in the strategic classification literature, as one cannot compare such access across mechanisms that accept different numbers of people overall.

 Overall, we believe that our strategic ranking model is a natural one through which to study constrained allocation settings for strategic decision making, such as admissions and hiring. While our base model is general, equity notions in other applications may differ from those considered in \Cref{sec:environment}. 
 From a technical perspective, there remain open computational questions for finding optimal $K$ level reward functions. 
More generally, our work supports the consideration of the ranking and constrained allocation model over unconstrained classification in the FATE ML community,  centering the study of welfare notions of fairness over that of statistical parity.




\subsubsection*{Acknowledgements}
The authors thank Frances Ding and John Miller for helpful comments on a draft. They also thank the anonymous reviewers for their invaluable feedback.


\bibliographystyle{abbrvnat}

\bibliography{mybib}

\newpage
 \appendix
 
 \thispagestyle{empty}

 \onecolumn \makesupplementtitle

\section{Supplemental model discussion}

Here we include supplemental technical remarks on the model.

\begin{remark}[Construction of the $\gamma$ map]\label{rem:gamma}
	Formally,  if $v_0$ is an atom of the score distribution and $\Omega_{v_0}$ is the set of tied applicants with score $v_0$,
	then the discontinuity of $F$ {at  $v_0$} has height equal to the measure of $\Omega_{v_0}$, and so $\gamma$ for those applicants can be filled by using the CDF of $\Gamma(\omega)$ restricted to
	$\Omega_{v_0}$; this gives a distribution for $\gamma(\omega, v_0)$ that is uniform over
	the gap interval $[ F_-(v_0),  F(v_0)]$ when restricted to
	$\Omega_{v_0}$, and hence leads to the claimed uniform distribution of $\theta_\post(\omega)$.
	Finally, we define $\gamma(\omega, v)$ for applicants  $\omega\notin\Omega_{v_0}$ by ``slotting them in'' in such a way that for a pair $\omega,\omega'$ with exactly one member in $\Omega_{v_0}$ and $\Gamma(\omega)<\Gamma(\omega')$ we have that
	$\gamma(v_0,\omega)\leq\gamma(v_0,\omega')$.
	Since the image of $\Omega_{v_0}$ is by construction dense in the gap interval, this uniquely determines $\gamma(v_0,\omega)$ for all applicants $\omega$.

	Note that for a given applicant $\omega$, the map $v\mapsto \gamma(\omega,v)$ is not necessarily 1-1; indeed, if $F$ is constant on an interval $I$, all $v\in I$ lead to the same rank.  But this only effects regions where the distribution of $v$ has no mass, and thus will not cause any issues.
\end{remark}

\begin{remark}[Ties in pre-effort skill]
	With probability one, there are no two applicants with the same skill level, and the support of the skill distribution has no gap: given two applicants with different skills, the probability of finding an applicant with skill in between these two is always non-zero.
\end{remark}

\begin{remark}[$\gamma$'s dependency on the set of efforts]\label{rem:gamma-effort}
	Formally, $\gamma$ depends on the set of efforts. Given a fixed set and $\gamma$, for each applicant $\omega$ the first condition considers the counter-factual ranking of $\omega$ with different post-effort values but using the same ranking function. As defined, $\gamma$ yields a uniform distribution of ranks with such measure $0$ changes. In Appendix \Cref{lem:deviationsmeasure0}, we further prove that two effort sets equal up to sets of measure $0$ induce the same ranking function $\gamma$, and so the condition is consistent.
\end{remark}

\begin{remark}[Credible commitment on the part of the school]\label{rem:design-opt}
	 Note that we do not require that, given applicant effort levels, the design $\lambda$ is optimal for the school's utility. As in e.g.~\citet{braverman2020role} for classification, the randomized ranking reward (characterized by probability $\lambda$) is not optimal for the school \textit{after} applicants have chosen their effort levels in response to the classifier. Such a characterization of equilibria thus requires credible commitment on the part of the school.
\end{remark}

\begin{remark}[Interpreting $\Util\pri$ as a conditional expectation.]\label{rem:prob}
	In our primary admissions interpretation, $\lambda(\theta_\post)$ is a probability. (All our results also hold when $\lambda(\theta_\post)$ represents a deterministic reward.) In a setting where there are a finite number of applicants (as opposed to our continuum model) and the admission outcome of each applicant with rank $\theta_\pre$ is $Z \sim Bernoulli(\lambda(\theta_\post))$, $\Util\pri$ can also be interpreted as a conditional expectation (average score of admitted applicants), that is $\E\left[\E[v\mid Z=1]\right]$.
\end{remark}

\begin{remark}[Comparison to related work in contests literature]\label{rem:contests}
 Proposition~\ref{lem:rankpreservedindex} has analogues in the contests literature, in particular \citet{bodoh2018college} and \citet{olszewski2016large}. Here we remark on the technical differences with these results. 	\citet{bodoh2018college} assumes that the density of prizes has full support---this assumption, translated to our setting, requires the ranking reward function $\lambda$ to be continuous. In contrast, Proposition~\ref{lem:rankpreservedindex} is proven for $\lambda$ that is a discontinuous step-function. Theorem 2(a) of \citet{olszewski2016large}, while not requiring $\lambda$ to be continuous, holds for all but a small fraction of the applicants. In other words, Proposition~\ref{lem:rankpreservedindex}, though recalls the assortative allocation principle known to the contests literature, is neither a restatement of previous results nor a generalization. It is worthwhile to note that we prove Proposition~\ref{lem:rankpreservedindex} using elementary arguments that may be of independent interest.

 \citet{olszewski2016large, olszewski2019pareto} also remark that the unique mechanism that implements assortative allocation is given by \citet{myerson1981optimal}, assuming quasi-linear utility.  
 However, due to the non-quasi-linearity of our utility function $U$ in the score $v(e, \theta_\pre)$, the equilibrium strategy in our setting (Theorem~\ref{lem:secondpricecontimuum}) turns out to be a variant of Myerson's payment rule, and is derived from an independent analysis.
\end{remark}

\begin{corollary}[Effort comparative statics]
	\label{corr:effectcomparativestatics}
	{Assume that $g$ and $\cost$ are differentiable.}
	{If} $\ell_k$ increases {and $\ell_j$ decreases for some $j>k$} (fixing all other parameters), then {$e_{i}(\theta)$ for all $i < k$ are unaffected,
		$e_k$ is weakly increasing, and efforts $e_{k+1}\dots$, $e_j$ are weakly decreasing.}
\end{corollary}


Perhaps surprisingly, increasing a reward $\ell_k$ does not affect the equilibrium effort of the applicants in the band $k - 1$ immediately below: in equilibrium, they do not receive the higher reward since those in band $k$ correspondingly increase their effort. 
The proof of the corollary actually implies that $e_k$, $e_{k+1}$ and $e_j$ are strictly monotonic
in the part of $\psi_k$, $\psi_{k+1}$ and $\psi_j$ where the efforts are strictly above $e_0$ --  which depending on the parameters of the model can just be part of these intervals, or all of these intervals, as observed in the paragraph following Theorem~\ref{lem:effort}.



\section{Multi-dimensional skill}
\label{sec:multi-dim}
Up to now, we have studied the setting in which there is a single-dimensional measurable score on which applicants are ranked and can exert effort. In this section we consider an extension of our model to $m$ skills, and use the extension to explore two questions of interest: (1) How does the school's reward design influence the applicant's decision to allocate effort across different skills in the competitive setting? (2) When one of the skills is valued but \emph{not} measurable by the school, how does competition in the measurable skill trade-off affect the school's utility? We first introduce the model extension, and then present results in Sections~\ref{sec:multi:linear} and~\ref{sec:multi:unmeasurable} addressing questions (1) and (2) respectively. We note that the results in this section are preliminary; we end the section with a discussion on directions for future inquiry. All proofs in this section can be found in Appendix~\ref{app:sec5}.


\parbold{Model.} The model is similar to our base model \Cref{sec:model}; each applicant now has $m$ latent skill levels with respective ranks $\theta_\pre^i \in [0,1]$ for $i\in \{1, \cdots, m\}$. Each rank is drawn independently from the uniform distribution over $[0,1]$, and the skill $i$ of applicant with rank $\theta_\pre^i$ is $\f{i}(\theta_\pre^i)$. Each applicant now chooses effort levels $\{ \e{i} \}$, at cost $p^m( \e{1}, \cdots, \e{m})$, resulting in post-effort scores $\{\score{i}\}$,
\begin{equation*}
\score{i} = g(\e{i})\cdot \f{i}(\theta_\pre^i),
\end{equation*}
where $g$ is concave, increasing, as before, and $\f{i}$ is a continuous, strictly increasing quantile function (${\f{i}}^{-1}$ is the CDF function of the scores on dimension $i$). As before, the school observes the post-effort scores for each applicant, now for each dimension $i$, and designs a non-decreasing function $\lambda : [0,1] \to [0,1]$ denoting the admissions probability $\lambda(\theta_\post)$ for applicant with post-effort rank $\theta_\post$. 


How does the school construct post-effort rank $\theta_\post$? In general, each applicant's \textit{combined} post-effort score may be any function of the scores on each dimension, $\{\score{i}\}$. Here, we assume the following linear score function. The school announces weights $\alpha = (\alpha_1, \cdots, \alpha_m) \in \Delta_{m-1}$ (denoting the $m$-dimensional simplex); the weights $\alpha$ represent the relative emphasis placed on each skill for the admissions decision. Then, each applicant is ranked according to their combined score, $\combscore := \sum_{i=1}^m \alpha_i\score{i}$, resulting in their combined post-effort rank $\theta_\post^\alpha$, and is admitted with probability $\lambda(\theta_\post^\alpha)$.\footnote{The linear combination of skill is similar to the ``linear mechanism'' in \citet{kleinberg18investeffort}, which studied how reward design incentivizes strategic agents to exert effort on different skill dimensions. Compared to their work, where efforts are connected to skills via an effort graph, we consider a simplified setting where each effort maps to one skill, and study how reward design affects equilibrium rankings, in presence of competition.}


Putting things together, the applicant's individual welfare is:
\begin{equation*}
\iwelf(\{ \e{i} \}_{i=1}^m, \lambda(\theta_\post^\alpha)) = \lambda(\theta_\post^\alpha) - p^m( \e{1}, \cdots, \e{m})
\end{equation*}

\subsection{Equilibrium under multi-dimensional competition}\label{sec:multi:linear}


In this section, we apply the model for multi-dimensional skills described above to study equilibrium effort allocations and rank preservation. 
%
%
%
In the single skill case,  we found that the post-effort ranks equaled the pre-effort ranks (the allocation preserves rank) 
(\Cref{lem:rankpreservedindex}). 
As we'll see, post-effort rank preservation 
is more complex with a multi-dimensional score. In fact, to our knowledge, the economics contests literature has not considered multi-dimensional scores, and defining assortative allocations (the analogue of rank preservation) in other domains (such as participant search in a matching market setting) has proven tricky~\cite{lindenlaub2016multidimensional}. Our first result for the multi-dimensional case is that rank preservation no longer holds if pre-effort ranks are defined as each applicant's pre-effort skill combination, $\alpha \cdot f(\theta_\pre)$. 
Rather, the post-effort ranks may depend on the weights $\alpha$ and the distribution of the skill levels, $\f{i}$; however, it may be possible to define an alternative pre-effort rank function under which rankings are preserved.

We show this result by characterizing the equilibrium of the following simplified setting, with further assumptions on the effort cost function $\cost$ and the effort transfer function $g$; the cost function $\cost$ is assumed to be a function of the sum of the efforts exerted: 
\[p(e^1, \cdots, e^m) \triangleq p\left(\sum_{i=1}^m \e{i}\right),\] where $p$ is convex and increasing.
This assumption says that the effort exerted for any skill is entirely exchangeable, for example, two hours spent on studying math is as costly as two hours spent on studying chemistry. 
Effort transfer function $g$ is assumed to be linear, that is, there are constant returns to effort.

Under these assumption, applicants are incentivized to put effort into a single skill; rankings are preserved not on the pre-effort skill combination but rather just their most important skill dimension. 

\begin{restatable}[Multi-dimensional rank preservation for linear $g$]{proposition}{propLinearG}
	Suppose $g$ is a
	linear function such that $g(e) = h x$, $h > 0$. 
	Suppose the school picks some $\alpha$ and $\lambda$. Define the combined pre-effort index as:
	\begin{align*}
		v_\pre^\alpha := \max_i \alpha_i\f{i}(\theta_\pre^i).
	\end{align*}
	Then in every equilibrium, for any two applicants with combined pre-effort indices $v_\pre^\alpha, \overline{v_\pre^\alpha}$ and  combined post-effort ranks $\theta_\post^\alpha, \overline{\theta_\post^\alpha}$ we have
	\begin{equation*}
	\lambda(\theta_\post^\alpha) >  \lambda(\overline{\theta_\post^\alpha}) \iff v_\pre^\alpha >  \overline{v_\pre^\alpha}. 
	\end{equation*}
	\label{prop:linearg}
\end{restatable}

Each applicant's decision to exert effort at equilibrium now also depends on the skill-specific quantile $\f{i}$, whereas in the single skill setting, only the pre-effort rank was relevant. Applicants who have a skill that they have a large advantage in \emph{relative} to the rest of the applicants (high $\f{i}(\theta_\pre^i)$) and that is valued highly by the school (high $\alpha_i$) are advantaged in terms of post-effort rankings.

For non-linear $g$, there may not exist a simple characterization of the post-effort ranks. For example, when $g$ is strictly concave, that is, there are decreasing marginal returns to effort for each skill, it is no longer optimal for a applicant to only put all their effort in one skill, and the relative allocation of effort will depend on the specific functional form of $g$.

\subsection{School's private utility with unmeasurable skill}\label{sec:multi:unmeasurable}

We now analyze a multi-dimensional setting where one skill dimension is \textit{unobservable} though still valuable to the school. This setting is directly motivated by the classic contract design work of~\citet{holmstrom1991multitask}, who show that when some work tasks are less measurable than others, it may be optimal to lessen incentives on the measurable tasks (e.g., by adopting a fixed wage) so as to not crowd out effort in the less measurable ones. 


In our simplified setting, there are two skills $M$ and $U$: $M$ has a measurable score $v^M$ and $U$ has an unmeasurable score $v^U$. For example, $M$ could be scholastic achievement as measured by SAT scores, and $U$ could be ``creativity", a personal quality that is valued by the school but is not directly measurable.
Since the school cannot observe $v^U$, its admission policy is based on $\theta_\post^M$ only, that is, $\alpha_U = 0$ and $\alpha_M = 1$.

We further assume that each applicant has a fixed effort budget of $B > 0$, and is intrinsically motivated to exert effort in the unmeasurable skill $U$; in fact, they will always exert effort $e^U = B - e^M$. Formally this corresponds to the effort cost function \[\cost^m(e_M, e_U) = \cost(e^M) - (\max(0,B-(e^M+e^U)))^2\] where $\cost$ is convex and increasing. We can write the applicant's individual welfare as:
\begin{equation*}
W(e^M, e^U, \lambda(\theta_\post^M)) = \lambda(\theta_\post^M) - \cost^m(e_M, e_U). 
\end{equation*}

The school's private utility is now weighted by $\beta \in (0,1)$, which quantifies the relative value the school places on the measurable skill over the unmeasurable skill.
\begin{equation*}
\Util\pri_\beta = \E[\beta\cdot v^M + (1-\beta)\cdot v^U\mid Z=1].
\end{equation*}
The smaller that $\beta$ is, the more the school places value on the unmeasurable skill $v^U$.


We now show that some degree of randomization in the admission policy may be optimal for the ranking designer if they value the unmeasurable skill sufficiently.

\begin{restatable}[The school's weighted private utility is maximized by some randomization]{proposition}{propWeightedUtil}\label{prop:2levels-fixedcap}
	Consider the class of two-level policies. For any $c \in (0,1-\rho)$, there exist some $\beta \in (0,1)$ such that the school's utility $\Util\pri_\alpha$ is maximized at that value of $c$.
\end{restatable}

The above result suggests that in the absence of the measurability of one of the skills, the non-randomized admissions policy may have added externality of promoting too much competition in the measurable skill, at the expense of the unmeasurable one. Our observation is in the same spirit as the finding by \citet{holmstrom1991multitask} that excessive performance-based incentives can lead strategic agents to focus only on dimensions of achievement that can be effectively measured. A difference between our setting and the aforementioned work is that the excess focus on the measurable skill is being driven by competition between applicants, not (directly) the design $\lambda$ of the school. However, the school can counter-act this competitive pressure and raise both its welfare and that of the applicants by adding randomization.


\paragraph{Discussion} In this section, we have only begun to explore strategic ranking in the multi-dimensional setting. We have outlined two promising directions of inquiry: \emph{multi-dimensional competition} and \emph{the lack of measurability}. For future work in the former direction, it would be natural to consider non-linear effort transfer functions, as well as the ramifications for disparate impact when the different environment factors are taken into account. In this setting, we have also highlighted $\alpha$ as an design choice---it determines how applicants tradeoff effort between the skill dimensions. Different stakeholders, such as the school and the society, may value the skills differently, due to, for example, long-term v.s. short-term considerations, and therefore have different preferences over $\alpha$. 

The direction regarding measurability is directly motivated by the classic contract design work of~\citet{holmstrom1991multitask}, who show that when some work tasks are less measurable than others, it may be optimal to lessen incentives on the measurable tasks (e.g., by adopting a fixed wage) so as to not crowd out effort in the less measurable ones. In our multidimensional setting, the college faces similar tradeoffs, with the additional challenge that it now must rank multiple applicants. 
For future work, it would be interesting to think about partial measurability and the effect of `weak' measurements on competition. In both cases, a core challenge is characterizing resulting equilibria and analoguous rank preservation results.

\section{Proofs for Section~\ref{sec:equillemmas}}

\begin{lemma}
In any equilibrium, tie-breaking 
{is not necessary}: ties in post-effort scores lead to ties in post-effort rewards. For any $\Gamma$, in any equilibrium the distribution of post effort scores $v$ is such that, for all $\omega, \omega'$
\[
v\left(e(\theta_\pre(\omega)),\theta_\pre(\omega)\right) = v\left(e(\theta_\pre(\omega')),\theta_\pre(\omega')\right) \triangleq v \implies
\lambda\left(\gamma\left(\omega, v \right)\right) = \lambda\left(\gamma\left(\omega', v \right)\right).
\]
\label{lem:tienotmatter}
\end{lemma}
\begin{proof} We prove the claim by contradiction. 
    Suppose not, and that $$v\left(e(\theta_\pre(\omega)),\theta_\pre(\omega)\right) = v\left(e(\theta_\pre(\omega')),\theta_\pre(\omega')\right) \triangleq v$$
    but $\lambda\left(\gamma\left(\omega, v \right)\right) < \lambda\left(\gamma\left(\omega', v \right)\right)$. Then, for any $\epsilon > 0$, we have that $ v\left(e(\theta_\pre(\omega)) + \epsilon,\theta_\pre(\omega)\right) > v$ (as the effort transfer function $g$ is strictly increasing) {and hence}
    \begin{align*}
        \lambda\left(\gamma\left(\omega, v \right)\right) < \lambda\left(\gamma\left(\omega', v \right)\right) &\leq \lambda\left(\gamma\left(\omega, v\left(e(\theta_\pre(\omega)) + \epsilon,\theta_\pre(\omega)\right) \right)\right).
    \end{align*}
  Since the function $p$ is continuous, we have that $p(e(\theta_\pre(\omega)) + \epsilon) - p(e(\theta_\pre(\omega))) \to 0$ as $\epsilon \to 0$. 
    
    Then, for small enough $\epsilon$, we have that
    \[ \lambda\left(\gamma\left(\omega, v\left(e(\theta_\pre(\omega)) + \epsilon,\theta_\pre(\omega)\right) \right)\right) - p(e(\theta_\pre(\omega)) + \epsilon) > \lambda\left(\gamma\left(\omega, v \right)\right) - p(e(\theta_\pre(\omega)))
     \]
    and thus the effort $e(\theta_\pre(\omega))$ is not welfare maximizing for the applicant $\omega$, a contradiction for it being the equilibrium effort for $\omega$. 
\end{proof}

\begin{lemma}[Deviations of measure 0]
\label{lem:deviationsmeasure0}
Fix $\Gamma$ and $\lambda$. Consider strategy set $\{e(\theta_\pre(\omega)\}$, and corresponding CDF $F$ of post-effort scores. As defined, $\gamma(\omega, v)$ is uniquely determined by $F$ and $\Gamma$.  
Now, suppose a measure $0$ set $\{\omega\}$ deviates, leading to strategy set $\{\tilde e(\theta_\pre(\omega)\}$, post-effort value distribution $\tilde F$, and ranking function $\tilde\gamma$.  

Then, $\tilde F=F$, and for all $\omega$ and $v$,
$$\lambda(\tilde\gamma(\omega, v)) =  \lambda(\gamma(\omega, v)).$$



\end{lemma}

\begin{proof}
By supposition, $\{e(\theta_\pre(\omega)\} = \{\tilde e(\theta_\pre(\omega)\}$ except at a set of measure $0$, and so the CDFs of the post-effort scores are equal, $F = \tilde F$. 
Now, recall that ranking function $\gamma(\omega, v)$ is defined as the CDF $F$ except where there are ties of positive mass (atoms) in the distribution of $v$.

For $v$ such that there is not an atom at $v$, the equality follows. 

If there is an atom at $v$, note that for all $\omega \in \{\omega : v(e(\theta_\pre(\omega)), \theta_\pre(\omega)) = v\}$, we have

$$ \gamma (\omega, v) \in [{\lim\sup}_{r \uparrow v} F(r),  F(v)] = [{\lim\sup}_{r \uparrow v} \tilde F(r),  \tilde F(v)]. $$
Finally, note that $\tilde \gamma (\omega, v) \in [{\lim\sup}_{r \uparrow v} \tilde F(r),  \tilde F(v)]$, and in particular the CDF of $\Gamma(w)$ restricted to the atomic set does not change due to measure $0$ changes. The equality follows. 
\end{proof}

\Cref{lem:deviationsmeasure0} characterizes the effect of measure $0$ deviations. While not needed for our results, it establishes that the behavior model in the equilibrium definition is consistent.


~\\\noindent\textbf{ Proposition~\ref{lem:rankpreservedindex}}: {\it
 In every equilibrium, $ \lambda(\theta_\post(\theta_\pre))=\lambda(\theta_\pre)$, up to {sets of}
measure 0. }

\begin{proof}

{We will prove that} in every equilibrium we have that for all $\omega, \omega'$, $\lambda(\theta_\pre) > \lambda(\theta_\pre') \implies \lambda(\theta_\post) {\geq} \lambda(\theta_\post').$ \footnote{Here we are using the shorthand notation $\theta_\pre:=\theta_\pre(\omega)$ and  $\theta_\pre':= \theta_\pre(\omega')$, as well as $\theta_\post := \theta_\post(\theta_\pre)$ and $\theta_\post' := \theta_\post(\theta_\pre')$.} 
{To see that this implies the claim, recall the definition of $\lambda$ as a function taking $K$ possible values $\ell_0,\dots,\ell_{K-1}$, and consider  
 $\omega_0,\dots,\omega_{K-1}$ such that $\lambda(\theta_\pre(\omega_k))$ runs through the $K$ possible values for $\lambda$ in increasing order.  The above inequality then implies that the order is weakly preserved if instead of $\lambda(\theta_\pre(\omega_k))$ we consider {$\lambda(\theta_\post(\theta_\pre(\omega_k)))$}. However, combined with the fact that $\theta_\post$ is uniformly distributed by our construction of the tie-breaking function, one easily shows the stronger statement that for almost all choices of $\omega_0,\dots,\omega_{K-1}$, the order is strictly preserved, which implies that claim of the proposition.}
  
 {We will prove the above monotonicity claim  by contradiction, and thus assume that}
there exists 
an equilibrium with a pair $\omega, \omega'$ with $\lambda(\theta_\pre) > \lambda(\theta_\pre')$  such that they spend effort efforts $e, e'$, respectively, to end up with post-effort skills $v, v'$ and $\lambda({\theta_\post}) <\lambda({\theta_\post'})$.  {The first inequality  implies that $f(\theta_\pre)>f(\theta_\pre')$  and the second implies} that
$v < v'$, via \Cref{lem:tienotmatter}. 
The high level proof idea {for why this is a contradiction}  is that, if it is worthwhile for $\theta_\pre'$ to spend effort $e' > e$ to reach skill $v'> v$ and thus reward $\lambda(\theta_{{\post}}')$, then it would also be worthwhile for $\theta_\pre$ to spend enough effort to reach skill $v'$ and reward $\lambda(\theta_{{\post}}')$, due to the convexity of the effort cost function $\cost(e)$.
	
	For convenience, in the proof we overload $\gamma$ to take in applicant effort as opposed to post-effort value as an argument:
	\[ \gamma(\omega, e) \triangleq \gamma(\omega, v(\theta_\pre(\omega), e)). \]
We also assume for convenience in the proof that for the relevant agents, the given effort levels are above the minimum effort level $e_0$. A near identical proof follows otherwise. 



	Then, from the definition of an equilibrium, we have:
	\begin{align}
		e&{\in} \argmax_{d} \left[ \lambda(\gamma(\omega, d)) - \cost(d) \right] & \text{for } \theta_\pre
		\label{eqnpart:argmax1}\\
		\theta_\post &= \gamma(\omega, e)\nonumber\\
		e'&{\in} \argmax_{d} \left[ \lambda(\gamma(\omega', d)) - \cost(d) \right] & \text{for } \theta_\pre' \label{eqnpart:argmax2}\\
		\theta_\post' &= \gamma(\omega', e')\nonumber
	\end{align}

	Holding all other effort levels fixed, let $\tilde e$ be the effort that $\theta_\pre$ would have needed to reach skill $v'$ (by the fact that the effort transfer function is continuous and strictly monontone, $\tilde e$ is uniquely determined by $v'$,
	{with  $e<\tilde e<e'$}) and thus rank reward $\lambda(\theta_\post')$ (due to \Cref{lem:tienotmatter} and the definition of $\gamma$, by deviating to reach a higher score $v'$, applicant $\omega$ receives the same rank reward as $\omega'$ does in the equilibrium with score $v'$). Similarly let $\tilde e'$, $e<\tilde e'<e'$ be the effort that $\theta_\pre'$ would have needed to reach score $v$ and thus rank reward $\lambda(\theta_\post)$. In other words, $\lambda(\gamma(\omega, e)) = \lambda(\gamma(\omega', \tilde e'))$, and $\lambda(\gamma(\omega', e')) = \lambda(\gamma(\omega, \tilde e))$. From Equations~\eqref{eqnpart:argmax1} and~\eqref{eqnpart:argmax2}, we have:
	\lledit{	\begin{align*}
		\lambda(\gamma(\omega, e)) - \cost(e) &{\geq }
		\lambda(\gamma(\omega, \tilde e)) - \cost(\tilde e) & \text{Eq. \eqref{eqnpart:argmax1}}
		\\\text{ and }\,\,\,\,\,\,\,\,\,\,\,\,\lambda(\gamma(\omega', e')) - \cost(e')
		&{\geq}\lambda(\gamma(\omega', \tilde e')) - \cost(\tilde e') & \text{Eq. \eqref{eqnpart:argmax2}}
		\end{align*}
		Applying Lemma~\ref{lem:order_p}, we get $\cost(e') - \cost(\tilde e') \leq \cost(\tilde e) - \cost(e). $
		In addition, Lemma~\ref{lem:order_g} implies  $e' - \tilde e' > \tilde e - e$.
}

	Thus, we have all of the following:
	\begin{align*}
		e' - \tilde e' &> \tilde e - e\\
		\cost(e') - \cost(\tilde e') &\leq \cost(\tilde e) - \cost(e) \\
		\tilde e' &> e
	\end{align*}
	However, they together contradict the assumption that $\cost{(x)}$ is convex and strictly increasing for $x \geq e_{{0}}$.
\end{proof}

\lledit{
\begin{lemma}\label{lem:order_p}
Suppose there are applicants $\omega, \omega'$ with $\theta_\pre(\omega') <\theta_\pre(\omega) $ and effort levels, $e<\{\tilde e', \tilde e \}<e'$ such that $\lambda(\gamma(\omega, e)) = \lambda(\gamma(\omega', \tilde e'))$, and $\lambda(\gamma(\omega', e')) = \lambda(\gamma(\omega, \tilde e))$. Then the following inequalities
\begin{align}
    	\lambda(\gamma(\omega, e)) - \cost(e) &{\geq }
		\lambda(\gamma(\omega, \tilde e)) - \cost(\tilde e) \label{eq:OPT1} \\
		\lambda(\gamma(\omega', e')) - \cost(e')
		&{\geq}\lambda(\gamma(\omega', \tilde e')) - \cost(\tilde e') \label{eq:OPT2}
\end{align}
imply $\cost(e') - \cost(\tilde e') \leq \cost(\tilde e) - \cost(e). $
\end{lemma}
\begin{proof}
    From Equations~\eqref{eq:OPT1} and~\eqref{eq:OPT2}, we have 	\begin{align*}
		\lambda(\gamma(\omega, e)) - \cost(e) &{\geq }
		\lambda(\gamma(\omega, \tilde e)) - \cost(\tilde e) & 
		\\
		\iff \lambda(\gamma(\omega, \tilde e)) - \lambda(\gamma(\omega, e))  
		&{\leq}
		\cost(\tilde e) - \cost(e) & \text{re-arrange}\\
		\\\text{ and }\,\,\,\,\,\,\,\,\,\,\,\,\lambda(\gamma(\omega', e')) - \cost(e')
		&{\geq}\lambda(\gamma(\omega', \tilde e')) - \cost(\tilde e') 
		\\
		\iff \lambda(\gamma(\omega, \tilde e)) - \cost(e') &{\geq}
		\lambda(\gamma(\omega, e)) - \cost(\tilde e') & \text{defn of } \tilde e, \tilde e'\\
		\iff \lambda(\gamma(\omega, \tilde e)) - \lambda(\gamma(\omega, e)) 
		&{\geq}
		\cost(e') - \cost(\tilde e') & \text{re-arrange}\\
	\end{align*}
	{Since} $e' > \{\tilde e', \tilde e\} > e$ {and $e\geq e_0$ in any equilibrium,}
	{we have that} 
	$\cost(e') > \{\cost(\tilde e'), \cost(\tilde e)\} > \cost(e)$. 
	Thus:
	\begin{equation*}
		0 < \cost(e') - \cost(\tilde e') \leq \lambda(\gamma(\omega, \tilde e)) - \lambda(\gamma(\omega, e)) \leq \cost(\tilde e) - \cost(e) 
	\end{equation*}
\end{proof}
\begin{lemma}\label{lem:order_g}
Suppose there are applicants $\omega, \omega'$ with $\theta_\pre(\omega') <\theta_\pre(\omega) $ and effort levels, $e<\{\tilde e', \tilde e \}<e'$ such that $v = v(e, \theta_\pre(\omega)) =v(\tilde{e}', \theta_\pre(\omega'))  $ and $v' =v(e', \theta_\pre(\omega'))=v(\tilde e, \theta_\pre(\omega))$. Then we have that $ e' - \tilde e' > \tilde e - e$.
\end{lemma}
\begin{proof}
    	By assumption, we have $g(e) = \frac{v}{f(\theta_\pre)}$,  $g(e') = \frac{v' }{f(\theta_\pre')}$,  $g(\tilde e) = \frac{v'}{f(\theta_\pre)}$, and $g(\tilde e') = \frac{v}{f(\theta_\pre')}$. Thus:
		\begin{align*}
			g(\tilde e) -g( e) &= \frac{v'  - v}{f(\theta_\pre)} \\
			&< \frac{v'  - v}{f(\theta_\pre')} & f(\theta_\pre') < f(\theta_\pre) \\
			&= g(e') - g(\tilde e').
		\end{align*}
		Since $g$ is concave and  increasing, and $\tilde e' > e$, we have that \[ g(e') - g(\tilde{e}') > g(\tilde{e})-g(e) \implies  e' - \tilde e' > \tilde e - e. \]
\end{proof}
}

\begin{proof}[Proof of Lemma~\ref{lem:effort}]
{We first  prove that the effort levels defined in the theorem lead to $\lambda(\theta_\pre) = \lambda(\theta_\post)$.  To this end, we first note that for $\theta_\pre\in\psi_k$ the efforts weakly decrease with $\theta_\pre$, while the score, $$v_k(\theta_\pre)=\max\{g(\tilde e_{k-1})f(c_k),g(e_0)f(\theta_\pre)\}$$  weakly increases with $\theta_\pre$. Thus
\begin{align*}
\min_{\theta_\pre\in\psi_k}v_k(\theta_\pre)&=
\max\{g(e_0), g(\tilde e_{k-1})\}f(c_k)
=g(\tilde e_{k-1})f(c_k)\\
&>g(e_{k-1}(c_k))f(c_k)
=\max\{g(e_0)f(c_k),g(\tilde e_{k-2})f(c_{k-1})\}
\\
&=v_{k-1}(c_k)=\sup_{\theta_\pre\in\psi_{k-1}}v_{k-1}(\theta_\pre).
\end{align*}
This shows that the scores attained by $\theta_\pre$ in $\psi_{k}$ are strictly larger than
those in $\psi_{k-1}$, which in turn implies that $\lambda(\theta_\pre) = \lambda(\theta_\post)$.}

Next we note that the effort levels from the theorem have been chosen to guarantee the following: (i) each applicant  $\theta_\pre\in \psi_k$ exerts the minimal effort that guarantees score $v_k(c_k)$ or higher, and (ii) 
for each applicant in $\psi_{k-1}$, the effort required to reach  score $v_k(c_k)$ or higher has an additional cost which is equal or larger than the benefit $\ell_k-\ell_{k-1}$.
 Note that these conditions are clearly necessary for the efforts to be an equilibrium -- but it is also sufficient for the following two reasons:
\begin{enumerate}
    \item The construction only directly prevents that an applicant with $\theta_\pre \in \psi_k$ would spend enough effort to match applicants with $\theta_\pre \in \psi_{k+1}$. However, we also have that they do not strictly benefit from matching applicants with $\theta_\pre \in \psi_{j}$ for some $j>k+1$.
    \lledit{This argument follows from a similar proof by contradiction, as the proof of \Cref{lem:rankpreservedindex}. For $j=k+2$: suppose applicants with $\theta_\pre \in \psi_k$ wish to match applicants with $\theta_\pre \in \psi_{k+2}$. By construction, applicants with $\theta_\pre \in \psi_{k+1}$ do not prefer to match applicants with $\theta_\pre \in \psi_{k+2}$ over staying in level $k+1$. Also, by construction, applicants with $\theta_\pre \in \psi_{k}$ do not wish to match applicants with $\theta_\pre \in \psi_{k+1}$ over staying in level $k$, but do prefer to move to level $k+2$. Then we may apply Lemmas~\ref{lem:order_p} and~\ref{lem:order_g} to obtain a contradiction. Now, iteratively apply the argument for each $j>k+2$.}

\item The construction also guarantees that those in band $\psi_k$ do not profit from reducing their effort, ending up in a lower band $\psi_j$, for $j\leq k-1$. First, consider $j=k-1$; this follows from strict convexity and monotonicity of $p$ in the relevant region, via an argument similar to that of \Cref{lem:rankpreservedindex}: if an applicant in $\psi_k$ benefits from moving to a lower band, then an applicant in the lower band strictly benefits to moving to band $\psi_k$, which is prevented by construction and the previous argument. More precisely, let $e'_h$ and $e'_l$ be the infimum over the effort levels that applicants in band $k-1$ need reach the higher and lower bands, respectively, i.e., the effort levels of the applicant with rank $c_k$. Let $e_h$, $e_l$ similarly be the requisite effort levels for some applicant in $\psi_k$ with $\theta_\pre > c_k$. Recall that $\ell_k \triangleq \lambda (\theta)$ for $\theta \in \psi_k$. By construction, we have 
\begin{equation}\label{eq:construct}
    \ell_k - \ell_{k-1} = p(e'_h) - p(e'_l).
\end{equation}

 As in \Cref{lem:rankpreservedindex}, we further have $e'_h - e'_l > e_h - e_l$, and $e'_l > e_l$. Since $p$ is strictly convex and increasing in the region, we have $p(e_h) - p(e_l) < p(e'_h) - p(e'_l)$. This fact together with \Cref{eq:construct} implies $\ell_k - \ell_{k-1} > p(e_h) - p(e_l)$, i.e., that the applicant in $\psi_k$ does not prefer to earn a lower score. \lledit{Now, for $j < k-1$, iteratively apply the same argument made above in part 1 with Lemmas~\ref{lem:order_p} and \ref{lem:order_g}, but with going to lower levels instead of higher levels.}
 


\end{enumerate}

{This proves that the effort levels defined in the theorem indeed describe an equilibrium.  To complete the proof, we need to prove the converse.  To the end, we first observe that by}
\Cref{lem:rankpreservedindex}, in every equilibrium rank rewards are preserved, $\lambda(\theta_\pre) = \lambda(\theta_\post)$.
	Thus, in equilibrium for every applicant with $\theta_\pre \in \psi_i$, we can write,
	\begin{align}
		e_i(\theta_\pre) &\in \argmax_{e} \left[ \lambda(\gamma(\omega, e)) - \cost(e) \right] & \text{Defn of equilibrium}\\
		\lambda(\theta_{\pre}) &= \lambda(\theta_{\post}) \triangleq \lambda(\gamma(\omega, e_i(\theta_\pre))) & \text{Lemma}~\ref{lem:rankpreservedindex}\nonumber
	\end{align}


	Consider a applicant with $\theta_\pre \in \psi_i$, for  $i >0$. Fix the effort levels for all applicants with $\theta_\pre' \in \psi_j$, at some $e_j(\theta_\pre')$, for $j < i$.
	%

	%
	Suppose $e_i(\theta_\pre) < g^{-1}\left(\frac{g(\tilde e_{i-1})\cdot f(c_i)}{f(\theta_\pre)}\right)$. Then, there exists a applicant with $\theta_\pre' \in \psi_j$, for $j=i-1$, who could strictly increase their rank and utility by changing their effort level to $e$ such that $ g^{-1}\left(\frac{g(e_i(\theta_\pre))\cdot f(\theta_\pre)}{f(\theta_\pre')}\right) < e < \tilde e_{i-1}$:
	\begin{align*}
		\cost(e) &= \ell_i - \cost(e) & g^{-1}\left(\frac{e_i(\theta_\pre)\cdot f(\theta_\pre)}{f(\theta_\pre')}\right) < e
		\\&>  \ell_{i-1} - \cost(e_{i-1}(\theta_\pre')) & e < \tilde e_{i-1}
	\end{align*}
	Thus, by contradiction, we have $e_i(\theta_\pre) \geq g^{-1}\left(\frac{g(\tilde e_{i-1})\cdot f(c_i)}{f(\theta_\pre)}\right)$.

	Now, consider an equilibrium set of effort levels $\{e_i(\theta_\pre)\}$, such that for all $i$ we have $e_i(\theta_\pre) \geq g^{-1}\left(\frac{g(\tilde e_{i-1})\cdot f(c_i)}{f(\theta_\pre)}\right)$. We note that it follows that $e_i(\theta_\pre) = \max\left(g^{-1}\left(\frac{g(\tilde e_{i-1})\cdot f(c_i)}{f(\theta_\pre)}\right), \argmin_e \cost(e)\right)$: for each applicant $\theta_\pre$, applicants in higher tiers exert enough effort such that $\theta_\pre$ does not wish to exert enough effort to overtake them, and tying them does not strictly increase their welfare. Thus, these applicants minimize their effort cost while still exerting enough effort to prevent lower tier applicants from overtaking them. Note that it is not an issue that $\theta_\pre$ is \textit{indifferent} in overtaking higher tiers. In any equilibrium, only a measure 0 set of such lower tier applicants will exert higher effort to tie the higher tier. Otherwise, tie-breaking will be necessary, in violation of \Cref{lem:tienotmatter}.    
	


	We finish the proof by showing there exists a unique effort set $\{e_i\}$ following the above. 
	For applicants in $\psi_0$, $\theta_{\pre} = \theta_{\post} \in \psi_0$, and so $e_0(\theta_{\pre}) = \argmin_e(\cost(e))$ in every equilibrium. Applicants do not need to spend any effort to stay at rank $0$, and in equilibrium they will not be incentivized to exert enough effort to increase their admissions probability. Recursively, efforts in higher tiers are fixed given effort in lower tiers, and so we have constructed the unique equilibrium.








\end{proof}


~\\\noindent \textbf{Proof for \Cref{corr:effectcomparativestatics}}
{Assume that $g$ and $\cost$ are differentiable.} 
	 {If} $\ell_k$ increases {and $\ell_j$ decreases for some $j>k$} (holding all other parameters fixed), then {$e_{i}(\theta)$ for all $i<k-1$ are unaffected,
	 		 $e_k$ is weakly increasing, and the efforts $e_{k+1}\dots$, $e_j$ are weakly decreasing.}

\begin{proof}
{By the inductive definitions of efforts, those for applicants $\theta_\pre\in \psi_i$ for $i<k$  do not change; to calculate the change of the effort for applicants in $\psi_k$, we differentiate the equation defining $\tilde e_{k-1}$, giving
$\cost'(\tilde e_{k-1})\frac d {d\ell_k}\tilde e_{k-1}=1$, and hence
$$
\frac {d \tilde e_{k-1}}{d\ell_k}=\frac 1{\cost'(\tilde e_{k-1})}>0.
$$
This implies that $\tilde e_{k-1}$ and hence $e_k$ is increasing in $\ell_k$.  Next we differentiate the defining equation for $\tilde e_{k}$ with respect to $\ell_k$, yielding
$$
\cost'(\tilde e_k)\frac {d\tilde e_k}{d\ell_k}
=
{p'(e_{k}(c_{k+1}))}\frac {de_k(c_{k+1})}{d\ell_k}-1
$$
 Using that
$g(e_k(c_{k+1}))=\max\left\{g(e_0), g(\tilde e_{k-1})\frac{f(c_k)}{f(c_{k+1})}\right\}$ we know that
if $\frac {de_k(c_{k+1})}{d\ell_k}\neq 0$, then
$$
g(e_k(c_{k+1}))=g(\tilde e_{k-1})\frac{f(c_k)}{f(c_{k+1})}< g(\tilde e_{k-1})
\quad\text{and}\quad
\frac {de_k(c_{k+1})}{d\ell_k}
   =\frac{f(c_k)}{f(c_{k+1})}\frac{g'(\tilde e_{k-1})}{g'(e_k(c_{k+1}))}\frac {d\tilde e_{k-1}}{d\ell_k}
$$
where the second equality follows from the first and the chain rule.
The first bound implies that
$e_k(c_{k+1})<\tilde e_{k-1}$ so by the concavity of $g$,
$g'(\tilde e_{k-1})\leq 
g'(e_k(c_{k+1}))$ and by the convexity of $p$,
$p'(\tilde e_{k-1})\geq 
p'(e_k(c_{k+1}))$; as a consequence
$$
\frac {de_k(c_{k+1})}{d\ell_k}<
  \frac {d\tilde e_{k-1}}{d\ell_k}=\frac 1{p'(\tilde e_{k-1})}
  \leq \frac 1{p'(e_{k}(c_{k+1}))}
  ,
$$
and hence
$$
\frac d{d\ell_k}\tilde e_k
=\frac{1}{p'(\tilde e_k)}\left(
{p'(e_{k}(c_{k+1}))}\frac {de_k(c_{k+1})}{d\ell_k}-1\right)
<0.
$$
{This shows that $e_{k+1}(\theta)$ is {weakly} decreasing in $\ell_k$.} 
By the inductive definitions, the efforts $e_{k+2},\dots e_{j-1}$
are non-increasing, while {$\tilde e_{j-1}$ is strictly decreasing, due to the change of $\ell_j$.  This in turn implies that $e_j$
is weakly decreasing in $\psi_j$.}
}
\end{proof}

\section{Proofs for Section~\ref{sec:private-util-max}}\label{app:private-util-max-proof}

\textbf{Proof of \Cref{prop:student-welfare}.}
\begin{proof}
	
 We prove this result under more general conditions where $e_0 \ge 0$.
 
	For a $K$-level policy, we can write the applicant welfare as follows:
	\begin{align*}
	\swelf &= \sum_{i=0}^{K-1}\E\left[ \ell_i - \cost(e_i(\theta_\pre)) \mid \theta_\pre \in \psi_i  \right] \\
	&=\rho - \sum_{i=0}^{K-1}\E\left[ \cost(e_i(\theta_\pre)) \mid \theta_\pre \in \psi_i  \right]  \\
	&\le \rho,
	\end{align*}
	since the term $\E\left[ \cost(e_i(\theta_\pre)) \mid \theta_\pre \in \psi_i  \right]\ge 0$ for any $i$ such that $0 \le i\le K-1$.
	
	The maximum $\swelf$ is uniquely attained by a two-level $\lambda$ where $c_1 = 0$, $e_i(\theta_\pre) =  e_0$, and so $\cost(e_i(\theta_\pre) ) = 0$ for all $\theta_\pre$.
	
	In the two-level policy class:
	\begin{align*}
	\swelf &= (1-c)\cdot \E\left[ \ell_1 - \cost( \max(g^{-1} ( \frac{g(\tilde e_0)f(c)}{f(\theta_\pre)} ) , e_0)) \mid \theta_\pre > c \right]\\
	&= \rho - \int_c^1  \cost( \max(g^{-1} ( \frac{g(\tilde e_0)f(c)}{f(\theta_\pre)} ) , e_0)) d\theta_\pre.
	\end{align*}
	Taking the derivative of the above display:
	\begin{align*}
	\frac{\partial \swelf}{\partial c} &= -\cost(\max(\tilde e_0, e_0)) - \int_c^1 \frac{\partial }{\partial c} \cost( \max(g^{-1} ( \frac{g(\tilde e_0)f(c)}{f(\theta_\pre)} ) , e_0)) d\theta_\pre
	\end{align*}
	which is non-positive since $\frac{\partial }{\partial c} \cost( \max(g^{-1} ( \frac{g(\tilde e_0)f(c)}{f(\theta_\pre)} ) , e_0)) \ge 0$.
\end{proof}

\noindent\textbf{Proof of \Cref{prop:private-util-max}.}
\begin{proof}[Proof]
 Note that for a two-level policy with cutpoint $c$, we have $\tilde e_0=\cost^{-1}(\ell_1) =\cost^{-1}(\frac{\rho}{1-c}) $. We have, since $g(e_0) = 0$,
 \begin{align*}
\E[v^M \mid Z=1] &= g(\tilde e_0) f(c)\\
&=g(\cost^{-1}(\frac{\rho}{1-c})) f(c)
 \end{align*} which is increasing in $c$. Thus $\Util\pri$ is maximized  by choosing $c$ to be as large as possible while filling the school's capacity, that is, $c = 1-\rho$, which corresponds to $\ell_1=1$.
 
\end{proof}

\begin{samepage}
\noindent\textbf{Proof of \Cref{prop-counterex-priv-util}.}
\begin{proof}[Proof]
Consider the following setting: The effort transfer function $g$ and effort cost function $\cost$ are such that $g(\cost^{-1}(x)) = x$ and $g(e_0) = 0$. We will show the following claim: 

``For any three-level $\lambda$ with levels $(\ell_0, \ell_1, \ell_2) = (0, x, 1)$ and cutpoints $0 < c_1 < c_2 < c_3 = 1$, there exists a skill distribution with quantile function $f$ such that $\Util\pri$ is higher under $\lambda$ than under the two-level policy corresponding to non-randomization."


To show that $\Util\pri$ is higher under the three-level policy, we show the following inequality: 
\begin{align}
& g(\tilde e_{0}) \cdot f(c_1)\cdot (c_{2} - c_1) \cdot x+ g(\cost^{-1}(1-x+\cost g^{-1}(\frac{g(\tilde e_0)\cdot f(c_1)}{f(c_2)}))) \cdot f(c_2)\cdot (1 - c_2) \nonumber \\
> ~& g(\cost^{-1}(1)) \cdot f(1-\rho)\cdot \rho, \label{eq:priv_util_3to2}
\end{align}
where $\tilde e_0 = \cost^{-1}(x)$. The LHS is nothing but $\Util\pri$ under a three-level policy, and the RHS, $\Util\pri$ under the two-level policy with cutpoint $1-\rho$.

Using $g(\cost^{-1}(x)) = x$, we can rewrite \eqref{eq:priv_util_3to2} as follows:
\begin{align}\label{eq:simple_priv_util_3to2}
x\cdot f(c_1)\cdot (c_2-c_1)\cdot x + ((1-x)f(c_2)+x\cdot f(c_1)) \cdot (1-c_2) > f(1-\rho) \cdot \rho.
\end{align}

The above display is true whenever $f(c_2)$ is large enough, that is 
\begin{equation*}
f(c_2) >  \frac{f(1-\rho) \cdot \rho - x^2\cdot f(c_1)\cdot (c_2-c_1)-x\cdot f(c_1) \cdot (1-c_2)}{ (1-x)(1-c_2)}.
\end{equation*}
\end{proof}

\noindent\textbf{Proof of \Cref{prop:SocUtil}}
\begin{proof}
	Given that $g(e_0) = 0$, we can rewrite the societal utility as
	\begin{align*}
		\Util\soc &= \E[v^M] \\
		&= (1-c)\cdot g(\tilde e_0)\cdot f(c),
	\end{align*}
	which attains the value $0$ at $c=0$ and $c=1$. Therefore it has an interior maximizer $c^*\in (0,1-\rho)$, if and only if 
	\begin{align*} \frac{\partial}{\partial c}(1-c)\cdot g(\cost^{-1}(\frac{\rho}{1-c}))\cdot f(c) < 0, \text{ at }c=1-\rho
	\end{align*}
	It suffices to check that the above inequality is true for
	\begin{equation*}
		f(x)=x, ~g(x) = \sqrt{x},~\cost(x)=x^2,~ \rho < 5/9.
	\end{equation*}
\end{proof}
\end{samepage}

\section{Comparison table for Section~\ref{sec:environment}}

\begin{table}[htbp!]
	\begin{center}
		\scalebox{0.8}{
			\begin{tabular}{| p{4.3cm} |  p{4.6cm}| p{4.6cm}| p{4.6cm}|} 
				\hline
				&  \thead{``Low'' region} & \thead{``Middle'' region}& \thead{``High'' region} \\ [0.05ex] 
				\makecell{Latent skill rank $\theta_\true$}  
				& $\makecell{0 \le \theta_\true < \thresA{c}}$ 
				& $\makecell{\thresA{c} \le \theta_\true < \thresB{c}}$ 
				& $\makecell{\thresB{c} \le \theta_\true \le 1 }$\\ 
				\hline
				\makecell{Welfare}   
				& $\makecell{\swelf^\A(\theta_\true) = \swelf^\B(\theta_\true) = 0}$ 
				& $\makecell{\swelf^\A(\theta_\true) \ge \swelf^\B(\theta_\true) = 0} $
				& \makecell{$\swelf^\A(\theta_\true) > \swelf^\B(\theta_\true) \ge 0$ } \\ 
				\hline
				\makecell{Admission\\ probability $\lambda(\theta_\post)$}   
				& \makecell{$\lambda^\A(\theta_\true) = $ \\ $ \lambda^\B(\theta_\true)= 0$} 
				&  \makecell{$ \lambda^\A(\theta_\true) = \frac{\rho}{1-c} > $ \\ $ \lambda^\B(\theta_\true)= 0$} 
				&  \makecell{$\lambda^\A(\theta_\true) = $ \\ $ \lambda^\B(\theta_\true) = \frac{\rho}{1-c}$}  \\ 
				\hline
				\makecell{Effort} 
				& \makecell{$e^\A(\theta_\true)= $ \\ $ e^\B(\theta_\true) = e_0$} 
				&  \makecell{$e^\A(\theta_\true)> $ \\ $ e^\B(\theta_\true) = e_0$} 
				& \makecell{$e^\B(\theta_\true)> $ \\ $ e^\A(\theta_\true) >  e_0$} \\  [.5ex] 
				\hline
		\end{tabular}}
	\end{center}
	\caption{Comparison of welfare, admission probability and effort between group $\A$ and $\B$ for every~$\theta_\true$, under a two-level policy parametrized by $c$. In the table we have used the shorthand notation $\lambda^\A(\theta_\true)$ for the admission probability $\lambda(\theta_\post(\theta_\true, \psi_\A))$ and $e^\A(\theta_\true)$ for the effort $e(\theta_\true, \psi_\A)$, and the respective notation for group $\B$.}\label{tab:comparison}
\end{table}

\section{Proofs for Section~\ref{sec:environment}}\label{app:environment}

\textbf{Proof of \Cref{prop:equi-env}}
\begin{proof}
	In order to transform the current setting into the setting of Proposition~\ref{lem:rankpreservedindex}, we absorb the environment factor $\psi$ into the applicant's latent skill level $f(\theta_\true)$ and compute the applicant's rank in $f(\theta_\true)\cdot \psi$. One can check that $f^{-1}_\mix(\cdot)$ as defined is indeed the appropriate CDF for the distribution of $f(\theta_\true)\cdot \psi$ in the overall applicant population. We can now write $v =g(e)\cdot f_{\mix}(\theta_\pre)$ and apply Proposition~\ref{lem:rankpreservedindex} to $\theta_\pre$.
\end{proof}

\noindent\textbf{Proof of \Cref{prop:environment-diff}}

\begin{proof}
	The first part of the result follows from Proposition~\ref{prop:equi-env}. Now to see that the statement about the welfare gap is true, we note the following: (1) For $0 \le \theta_\true < \thresA{c}$, both the group $\A$ and group $\B$ applicant have zero welfare because they have zero probability of admission and incur zero cost of effort; (2) For $\thresA{c} \le \theta_\true < \thresB{c}$, the group $\B$ has zero welfare (since they have zero probability of admission) whereas the group $\A$ with the same $\theta_\true$ has postive probability of admission and non-negative welfare; (3) For $\thresB{c} \le \theta_\true \le 1$, the group $\A$ and group $\B$ applicant with same $\theta_\true$ both have the same probability of admission but the group $\A$ applicant exerts less effort. This is because a applicant with environment-scaled rank $\theta_\pre$ exerts effort
	\begin{align}
		e = \begin{cases}
			g^{-1} \left( \frac{g(\tilde e_0) \cdot f_\mix(c)}{f_\mix(\theta_\pre)} \right)& \text{ if } \theta_\pre > c \\
			e_0 & \text{o.w.}
		\end{cases}
	\end{align}
	
	For the same $\theta_\true$, the group $\A$ applicant has higher $\theta_\pre$, since
	\begin{align*}
		f^{-1}_\mix(f(\theta_\true)\cdot \psi_\A) > f^{-1}_\mix(f(\theta_\true)\cdot \psi_\B).
	\end{align*}
	Since $e$ is decreasing in $\theta_\pre$, the group $\A$ applicant exerts less effort than the group $\B$ applicant.
\end{proof}

\noindent\textbf{Proof of \Cref{prop:deriv-welfare-gap}}
\begin{proof}
	We have that
	\begin{align*}
		\frac{\partial \welfgap(\theta_\true)}{\partial c} & =\frac{\partial }{\partial c} \swelf^\A(\theta_\true) - \frac{\partial }{\partial c}\swelf^\B(\theta_\true)
	\end{align*}
	
	Since $c < (f_\mix)^{-1}(f(\theta_\true)\cdot \psi_\B)$, we have that $\theta_\true > \thresB{c}$. Thus we may compute $\swelf^\A(\theta_\true)$ and $\frac{\partial }{\partial c}\swelf^\B(\theta_\true)$  as follows:
	\begin{align*}
		\swelf^\A(\theta_\true) &= \frac{\rho}{1-c} - \cost(g^{-1} ( \frac{g(\tilde e_0)f_\mix(c)}{f(\theta_\true)\cdot \psi_\A} )  ), \\
		\swelf^\B(\theta_\true) &= \frac{\rho}{1-c} - \cost(g^{-1} ( \frac{g(\tilde e_0)f_\mix(c)}{f(\theta_\true)\cdot \psi_\B} )  ),
	\end{align*}
	where as before $\tilde e_0 = \cost^{-1}(\frac{\rho}{1-c})$.
	Then we have:
	\begin{align*}
		\frac{\partial }{\partial c} \swelf^\A(\theta_\true) &= \frac{\rho}{(1-c)^2} \\
		&\quad -\cost'( g^{-1} ( \frac{g(\tilde e_0)f_\mix(c)}{f(\theta_\true)\cdot \psi_\A} ) )\cdot (g^{-1})' ( \frac{g(\tilde e_0)f_\mix(c)}{f(\theta_\true)\cdot \psi_\A} ) \cdot\frac{g(\tilde e_0)(f_\mix)'(c) + f_\mix(c)g'(\tilde e_0) \frac{\partial \tilde e_0}{\partial c}}{f(\theta_\true)\cdot \psi_\A}
	\end{align*}
	We claim that $\frac{\partial }{\partial c} \swelf^\A(\theta_\true) > \frac{\partial }{\partial c}\swelf^\B(\theta_\true) $. We have the following facts:
	\begin{align}
		&\frac{1}{\psi_\A}<\frac{1}{\psi_\B} \label{eq:gamma-ineq}\\
		&\cost '(x) \le \cost '(\tilde{x}) ~\forall x < \tilde{x} \label{eq:cost-convex}\\
		&g^{-1}(x) < g^{-1}(\tilde x)~\forall x < \tilde{x} \label{eq:g-increasing}\\
		&(g^{-1})'(x) \le (g^{-1})'(\tilde x)~\forall x < \tilde{x} \label{eq:g-concave}\\
		&g(x), g'(x), f_\mix(x), f(x) > 0~\forall x > 0\label{eq:g-f-positive}\\
		&(f_\mix)'(c) > 0 \label{eq:fmix-increasing}\\
		&\frac{\partial \tilde e_0}{\partial c} > 0 \label{eq:te0-increasing}\\
	\end{align}
	Equations \eqref{eq:gamma-ineq}, \eqref{eq:cost-convex}, \eqref{eq:g-increasing} and \eqref{eq:g-f-positive} imply
	\begin{equation}\label{eq:int-1}
	\cost '( g^{-1} ( \frac{g(\tilde e_0)f_\mix(c)}{f(\theta_\true)\cdot \psi_\A} ) ) \le \cost '( g^{-1} ( \frac{g(\tilde e_0)f_\mix(c)}{f(\theta_\true)\cdot \psi_\B} ) )
	\end{equation}
	Equations \eqref{eq:gamma-ineq}, \eqref{eq:g-concave} and \eqref{eq:g-f-positive} imply
	\begin{equation}\label{eq:int-2}
	(g^{-1})' ( \frac{g(\tilde e_0)f_\mix(c)}{f(\theta_\true)\cdot \psi_\A} ) \le (g^{-1})' ( \frac{g(\tilde e_0)f_\mix(c)}{f(\theta_\true)\cdot \psi_\B} )
	\end{equation}
	Equations \eqref{eq:gamma-ineq}, \eqref{eq:g-f-positive},  \eqref{eq:te0-increasing} and \eqref{eq:fmix-increasing} imply
	\begin{equation}\label{eq:int-3}
	\frac{g(\tilde e_0)(f_\mix)'(c) + f_\mix(c)g'(\tilde e_0) \frac{\partial \tilde e_0}{\partial c}}{f(\theta_\true)\cdot \psi_\A}  < \frac{g(\tilde e_0)(f_\mix)'(c) + f_\mix(c)g'(\tilde e_0) \frac{\partial \tilde e_0}{\partial c}}{f(\theta_\true)\cdot \psi_\B}
	\end{equation}
	Intermediate equations \eqref{eq:int-1}, \eqref{eq:int-2} and \eqref{eq:int-3} together imply that $\frac{\partial }{\partial c} \swelf^\A(\theta_\true) > \frac{\partial }{\partial c}\swelf^\B(\theta_\true) $. This gives $  \frac{\partial \welfgap(\theta_\true)}{\partial c} > 0$ as desired.
\end{proof}

\noindent\textbf{Proof of \Cref{prop:access}}
\begin{proof}
	For a two-level policy with cutpoint $c >0$, we have, by definition:
	\begin{align*}
		\access &= \frac{\rho}{1-c}\left( 1 - \thresB{c}\right). 
	\end{align*}
	Since $\thresB{c}  > c$ for any $c\in (0,1-\rho]$, we have that $\access < \rho$. On the other hand, $\access = 0$ for a pure randomization policy.
	
	Taking derivative of $\access$ with respective to $c$,
	\begin{align*}
		\frac{\partial\access  }{\partial c} &=  \frac{\rho}{1-c}\left(-\frac{\partial }{\partial c}\thresB{c}\right) + \frac{\rho}{(1-c)^2} \left( 1 -  \thresB{c} \right)\\
		&= \frac{\rho}{1-c}\left(-\frac{\partial }{\partial c}\thresB{c} + \frac{1-\thresB{c} }{1-c} \right). 
	\end{align*}
	We want to show that the above derivative is non-positive under the assumption that $f^{-1}$ is convex.
	Since $\frac{1-\thresB{c} }{1-c} \le 1$ (because $\thresB{c}  \le c$ for any $c\in [0,1-\rho]$), it suffices to show that $\frac{\partial }{\partial c}\thresB{c} \ge 1$. Indeed, we may compute
	\begin{align*}
		\frac{\partial }{\partial c}\thresB{c} &= \frac{(f^{-1})'\left(\frac{f_\mix(c)}{\psi_\B}\right) }{\psi_\B} \cdot \left(  \frac{1}{2}\cdot \frac{(f^{-1})'\left(\frac{f_\mix(c)}{\psi_\A}\right) }{\psi_\A} + \frac{1}{2} \cdot \frac{(f^{-1})'\left(\frac{f_\mix(c)}{\psi_\B}\right) }{\psi_\B}   \right)^{-1} \\
		&\ge 1,
	\end{align*}
	since $\frac{(f^{-1})'\left(\frac{f_\mix(c)}{\psi_\B}\right) }{\psi_\B} \ge \frac{(f^{-1})'\left(\frac{f_\mix(c)}{\psi_\A}\right) }{\psi_\A}$ by the convexity of $f^{-1}$ (which implies $(f^{-1})'$ is non-decreasing) and the fact that $\psi_\A > \psi_\B$. 
	
\end{proof}


\section{Proofs for Section~\ref{sec:multi-dim}}\label{app:sec5}
~	
%
\noindent\textbf{Proof of \Cref{prop:linearg}}
	\begin{proof}
		The students are sorted according to their weighted scores:
		\begin{align*}
			\combscore &= \sum_{i=1}^m \alpha_i g(\e{i}) \cdot \f{i}(\theta_\pre^i) \\
			&=h \sum_{i=1}^m  \e{i} \cdot \left(\alpha_i\f{i}(\theta_\pre^i) \right).
		\end{align*}
		Consider a single applicant. Denote $i^* := \argmax_{i} \alpha_i\f{i}(\theta_\pre^i)$. Since $\f{i}$'s are strictly monotone and $\theta_\pre^i$ are independently distributed, $i^*$ is almost everywhere unique. For any $e > 0$, let $\e{i^*}^* = e$ and $\e{i}^* = 0$ for each $i \ne i^*$. Then, holding the effort levels of all other applicants fixed, we have
		\begin{align*}
			\iwelf(\{ \e{i}^* \}_{i=1}^m, \lambda(\theta_\post^\alpha)) \ge \iwelf(\{ \e{i}\}_{i=1}^m, \lambda(\theta_\post^\alpha)) ~\forall \{ \e{i}\}_{i=1}^m ~s.t. \sum_{i=1}^m e^i = e.
		\end{align*}
		In other words, for any fixed $\sum_{i=1}^m e^i = e$, the applicant maximizes their welfare by putting all effort into the skill $i^*$ with the highest $\alpha^i\theta_\pre^i$. Therefore, it suffices for the applicant to maximize their individual welfare over the total effort level $e$. At equilibrium, the weighted score $\combscore$ satisfies \[\combscore = h \cdot e \cdot v_\pre^\alpha.\]
		
		Since $\lambda$ is non-decreasing in the weighted score $\combscore$, we may now retrace the proof of Proposition~\ref{lem:rankpreservedindex} to argue the following: if a student with $v_\pre^\alpha$ finds it optimal to achieve weighted score $\combscore$, then each student with $\overline{v_\pre^\alpha}$ finds it optimal to reach weighted score  $\overline{\combscore}\ge\combscore$, due to the convexity of the effort cost function $\cost(e)$.
	\end{proof}

\noindent\textbf{Proof of \Cref{prop:2levels-fixedcap}}
\begin{proof}
	We prove the result for $g(e_0)  = 0$ to simplify the presentation. The generalization to $g(e_0)$ can be accomplished by some additional book-keeping without modifying the core idea.
	
	Note that in this setting we have $\tilde e_0=\cost^{-1}(\ell_1) =\cost^{-1}\left(\frac{\rho}{1-c}\right) $.
	
	For the measurable skill:
	\begin{align*}
		\E[v^M \mid Z=1] &= g(\tilde e_0)\cdot f(c).
	\end{align*}

	For the unmeasurable skill:
	\begin{align*}
		\E[v^U \mid Z=1] &=  \E[g(B-e_1(\potM_\pre))\cdot f(\potM_\pre) \mid Z=1]\\
		&= \E\left[g(B-g^{-1}\left( \frac{g(\tilde e_0)f(c)}{f(\potM_\pre)} \right))\cdot f(\potU_\pre) \mid \potM_\pre > c\right]\\
		&= \E[\potU_\pre]\cdot \int_c^1 g(B-g^{-1}\left( \frac{g(\tilde e_0)f(c)}{f(\potM_\pre)} \right))d\potM_\pre \tag{independence of $\potU_\pre$ and $\potM_\pre$}
	\end{align*}
	
	We know that $\frac{\partial}{\partial c}\E[v^M \mid Z=1] > 0$ because $\tilde e_0$ is increasing in $c$.
	
	By the Leibniz integral rule, we have
	\begin{align*}
		\frac{\partial}{\partial c}  \E[v^U \mid Z=1] &= \E[\potU_\pre] \cdot \left( - g(B-\tilde e_0) + \int_c^1  \frac{\partial}{\partial c}  g(B-g^{-1}\left( \frac{g(\tilde e_0)f(c)}{f(\potM_\pre)} \right))d\potM_\pre \right)< 0,
	\end{align*}
	since $g(B-g^{-1}\left( \frac{g(\tilde e_0)f(c)}{f(\potM_\pre)} \right))$ is decreasing in $c$.
	
	For any $c \in (0, 1-\rho)$, set \begin{equation*}
		\alpha = \frac{-\frac{\partial}{\partial c}  \E[v^U \mid Z=1]}{\frac{\partial}{\partial c}\E[v^M \mid Z=1]-\frac{\partial}{\partial c}  \E[v^U \mid Z=1]}.
	\end{equation*}
	Then we have that \[  \frac{\partial}{\partial c} \Util = \alpha \frac{\partial}{\partial c}\E[v^M \mid Z=1] + (1-\alpha)  \frac{\partial}{\partial c}  \E[v^U \mid Z=1]=0. \] That is, the school's utility $\Util$ for the chosen $\alpha$ is maximized at an interior skill of $c$.
	
\end{proof}

\bibliographystyle{abbrvnat}

\bibliography{mybib}

\vfill

\end{document}